%% file: main.tex
\documentclass[11pt]{article}
\input{header.tex}
\usepackage{graphicx}
\usepackage{multirow}
\usepackage[T1]{fontenc}

\begin{document}

\linespread{0.96}

\date{}

\title{Space-Optimal Majority in Population Protocols }

\author{
  Dan Alistarh\thanks{Dan Alistarh is supported by a Swiss National Science Foundation Ambizione Fellowship.}\\
  \small ETH Zurich\\
	\small dan.alistarh@inf.ethz.ch\\
  \and
  James Aspnes\thanks{James Aspnes is supported in part by NSF grants CCF-1637385 and CCF-1650596.}\\
        \small Yale\\
        \small james.aspnes@yale.edu 
  \and
  Rati Gelashvili\thanks{Rati Gelashvili is supported by the National Science Foundation under grants CCF-1217921, CCF-1301926, and IIS-1447786, the Department of Energy under grant ER26116/DE-SC0008923, and Oracle and Intel corporations.}\\
        \small MIT\\
	\small gelash@mit.edu
}

\maketitle
\begin{abstract}
{\small
Population protocols
  are a popular model of distributed computing, 
  in which $n$ agents with limited local state interact randomly, 
  and cooperate to collectively compute global predicates.
Inspired by recent developments in DNA programming, an extensive series of papers, 
  across different communities, has examined the computability and complexity 
 characteristics of this model. 
Majority, or consensus, is a central task in this model, in which agents need to collectively reach 
  a decision as to which one of two states $A$ or $B$ had a higher initial count. 
Two complexity metrics are important: the \emph{time} that a protocol requires to stabilize 
  to an output decision, and the \emph{state space size} that each agent requires to do so. 

It is currently known that majority requires $\Omega( \log \log n )$ 
  states per agent to allow for fast (poly-logarithmic time) stabilization, 
  and that $O( \log^2 n )$ states are sufficient. 
Thus, there is an exponential gap between the upper and lower bounds for this problem. 

We address this question.
On the negative side, we provide a new lower bound of $\Omega( \log n )$ states for 
  any protocol which stabilizes in $O( n^{1-c} )$ expected time, for any constant $c > 0$. 
This result is conditional on basic monotonicity and output assumptions, satisfied by all known protocols. 
Technically, it represents a significant departure from previous lower bounds, 
  in that it does not rely on the existence of dense configurations. 
Instead, we introduce a new generalized surgery technique to prove the existence of incorrect executions 
  for any algorithm which would contradict the lower bound.
Subsequently, our lower bound applies to more general initial configurations.

On the positive side, we give a new algorithm for majority which uses $O( \log n )$ states, and stabilizes in $O( \log^2 n )$ expected time. 
Central to the algorithm is a new \emph{leaderless phase clock} technique, which allows agents to synchronize in phases of $\Theta(n \log{n})$ consecutive interactions using $O( \log n )$ states per agent, 
exploiting a new connection between population protocols and power-of-two-choices load balancing mechanisms.
We also employ our phase clock to build a leader election algorithm with 
  a state space of size $O( \log n )$, which stabilizes in $O( \log^2 n )$ expected time. 
}
\end{abstract}


\thispagestyle{empty}

\newpage
\setcounter{page}{1}

\input{intro.tex}
\input{model.tex}
\input{majlower.tex}
\input{sync.tex}
\input{majupper.tex}

\section{Conclusions and Future Work}
We have given tight logarithmic upper and lower bounds for the space (state) complexity 
  of fast exact majority in population protocols. 
Our lower bound is contingent on an output dominance assumption, 
  satisfied by all known protocols. 
Together with the recent $O(\log \log n)$-states leader election 
  by G{\k{a}}sieniec and Stachowiak~\cite{GS17}, our results suggest an exponential 
  gap between the space complexity of these two fundamental tasks in population protocols.
Unlike~\cite{DS15, AAEGR17}, we do not require fast stabilization from configurations 
  where all states have large counts.
As a result, our lower bound technique works from more general initial configurations,
  and thus is potentially applicable to broader settings.
It also applies to predicates such as equality. 
Similarly, the leaderless phase clock we introduce is quite general, 
  and should be applicable broadly. 
In particular, recent results~\cite{PTW15} suggest that it should be applicable 
  to settings where the communication graph can be modelled by an expander.  
Exploring and characterizing these generalizations is an interesting direction
  for future research.

One open question and a technical challenge that we would love to see settled 
  is getting rid of the ``output dominance'' assumption in our majority lower bound. 
We conjecture that the same $\Omega(\log n)$ lower bound must hold unconditionally.
Moreover, we conjecture that all majority algorithms that stabilize to the 
  correct decision in polylogarithmic time must satisfy ``output dominance'',
  which would obviously imply the unconditional lower bound.
Some of the technical tools developed in this paper will hopefully be helpful,
  but the proof of this conjecture is likely to at least involve more
  complex ``surgeries'' on transition sequences.

Our lower bound of $\Omega(\log n)$ states follows when the initial discrepancy
  between the counts of majority and minority states is not larger than $\sqrt{n}$.
We could ask what happens for larger discrepancies.
We think that in this case, it might be possible to stabilize in polylogarithmic
  time using $O(\log \log n)$ states.
One idea is to use truncated phased majority algorithm with less levels.

Our lower bound applies to algorithms that stabilize fast, which is a stronger 
  requirement than convergence. 
Our results highlight a separation. 
While fast stabilization for majority requires $\Omega(\log n)$ states, 
  it is possible converge fast using $O(\log \log n)$ states.
While some of our technical tools may well turn out to be useful when dealing
  with the convergence requirement as opposed to stabilization, 
  solving this problem might require developing 
  novel and interesting techniques.

\section{Acknowledgments}
We would like to thank David Doty, David Soloveichik, Petra Berenbrink
  and anonymous reviewers for helpful feedback.   

\bibliographystyle{alpha}
\bibliography{biblio}
\appendix
\input{appmajlower.tex}
\input{appsync.tex}

\input{appmajupper.tex}
\input{leupper.tex}

\end{document}

%% file: header.tex
\usepackage[letterpaper,margin=.97in]{geometry}
\usepackage{amsmath, amssymb, amsthm, amsfonts}
\usepackage{centernot}
\usepackage{ifthen}
\usepackage{arrayjobx}
\usepackage{tikz}
\usetikzlibrary{positioning,decorations.pathreplacing}
\usepackage[normalem]{ulem}
\usepackage{cite}
\usepackage{times}
\usepackage{appendix}
\usepackage{graphicx}
\usepackage{color}
\usepackage[noend]{algpseudocode}
\usepackage{epstopdf}
\usepackage{wrapfig}
\usepackage[framemethod=tikz]{mdframed}
\usepackage[bottom]{footmisc}
\usepackage{enumitem}
\setitemize{noitemsep,topsep=0pt,parsep=0pt,partopsep=0pt}
\usepackage{caption}
\usepackage{xspace}
\usepackage{hyperref}
\hypersetup{
    unicode=false,          
    colorlinks=true,        
    linkcolor=red,          
    citecolor=green,        
    filecolor=magenta,      
    urlcolor=cyan           
}
\usepackage[capitalize]{cleveref}
\usepackage{fancyhdr} 
\usepackage{color} 
\usepackage[nofillcomment,linesnumbered,noend,noresetcount,noline]{algorithm2e}
\usepackage[noindentafter]{titlesec}

\titlespacing{\paragraph}{%
  0pt}{
  0.1\baselineskip}{
  1em}
\titlespacing\section{0pt}{8pt plus 1pt minus 1pt}{2pt plus 1pt minus 1pt}
\titlespacing\subsection{0pt}{8pt plus 1pt minus 1pt}{2pt plus 1pt minus 1pt}
\titlespacing\subsubsection{0pt}{8pt plus 1pt minus 1pt}{2pt plus 1pt minus 1pt}
\usepackage{graphicx} 
\usepackage{complexity} 
\usepackage{enumitem} 
\usepackage{shorttoc} 
\usepackage{array} 
\usepackage{float}
\usepackage{pdfsync}
\usepackage{verbatim}
\hypersetup{pageanchor=false}

\AtBeginDocument{%
 \abovedisplayskip=2pt plus 1pt minus 1pt
 \abovedisplayshortskip=0pt plus 1pt
 \belowdisplayskip=2pt plus 1pt minus 1pt
 \belowdisplayshortskip=0pt plus 1pt
}

\newtheoremstyle{slplain}
  {.4\baselineskip\@plus.1\baselineskip\@minus.1\baselineskip}
  {.3\baselineskip\@plus.1\baselineskip\@minus.1\baselineskip}
  {\itshape}
  {}
  {\bfseries}
  {.}
  { }
  {}
\theoremstyle{slplain} 

\algnewcommand\algorithmicswitch{\textbf{switch}}
\algnewcommand\algorithmiccase{\textbf{case}}

\algdef{SE}[SWITCH]{Switch}{EndSwitch}[1]{\algorithmicswitch\ #1\ \algorithmicdo}{\algorithmicend\ \algorithmicswitch}%
\algdef{SE}[CASE]{Case}{EndCase}[1]{\algorithmiccase\ #1}{\algorithmicend\ \algorithmiccase}%
\algtext*{EndSwitch}%
\algtext*{EndCase}%

\renewcommand{\paragraph}[1]{\vspace{0.07cm}\noindent {\bf #1}:}
\newtheorem*{theorem*}{Theorem}
\newtheorem{theorem}{Theorem}[section]
\newtheorem{lemma}[theorem]{Lemma}

\newtheorem{claim}[theorem]{Claim}
\newtheorem{corollary}[theorem]{Corollary}
\newtheorem{definition}[theorem]{Definition}

\newtheorem{invariant}[theorem]{Invariant}

\makeatletter
\newtheorem*{rep@theorem}{\rep@title}
\newcommand{\newreptheorem}[2]{%
\newenvironment{rep#1}[1]{%
 \def\rep@title{#2 \ref{##1}}%
 \begin{rep@theorem}}%
 {\end{rep@theorem}}}
\makeatother

\newreptheorem{theorem}{Theorem}
\newreptheorem{lemma}{Lemma}
\newreptheorem{claim}{Claim}
\newreptheorem{invariant}{Invariant}
\newreptheorem{corollary}{Corollary}

\theoremstyle{definition}

\theoremstyle{remark}

\numberwithin{equation}{section}

\newtheoremstyle{etplain}
  {.0\baselineskip\@plus.1\baselineskip\@minus.1\baselineskip}
  {.0\baselineskip\@plus.1\baselineskip\@minus.1\baselineskip}
  {\itshape}
  {}
  {\bfseries}
  {.}
  { }
  {}

\newcommand{\idlow}[1]{\mathord{\mathcode`\-="702D\it #1\mathcode`\-="2200}}
\newcommand{\id}[1]{\ensuremath{\idlow{#1}}}
\newcommand{\litlow}[1]{\mathord{\mathcode`\-="702D\sf #1\mathcode`\-="2200}}
\newcommand{\lit}[1]{\ensuremath{\litlow{#1}}}

\newcommand{\namedref}[2]{\hyperref[#2]{#1~\ref*{#2}}}

\newcommand{\sectionref}[1]{\namedref{Section}{#1}}

\newcommand{\theoremref}[1]{\namedref{Theorem}{#1}}

\newcommand{\figureref}[1]{\namedref{Figure}{#1}}
\newcommand{\figurerefb}[2]{\hyperref[#1]{Figure~\ref*{#1}#2}}

\newcommand{\lemmaref}[1]{\namedref{Lemma}{#1}}

\newcommand{\invariantref}[1]{\namedref{Invariant}{#1}}

\newcommand{\corollaryref}[1]{\namedref{Corollary}{#1}}

\newcommand{\equationref}[1]{\hyperref[#1]{(\ref*{#1})}}
\renewcommand{\eqref}{\equationref}

\newcommand{\confspacek}{\Lambda_k \to \mathbb{N}}

\newcommand{\reach}{\Longrightarrow}

\usepackage[textsize=tiny]{todonotes}

\newcommand{\DEBUG}[1]{}

\renewcommand{\setminus}{-}
\renewcommand{\emptyset}{\varnothing}

\newcommand{\FullOrShort}{short}

\ifthenelse{\equal{\FullOrShort}{full}}{
        
  \newcommand{\fullOnly}[1]{#1}
  \newcommand{\shortOnly}[1]{}
  
  }{

    \newcommand{\fullOnly}[1]{}
    \newcommand{\shortOnly}[1]{#1}
    
  }



%% file: intro.tex
\section{Introduction}
Population protocols~\cite{AADFP06} are a model of distributed computing 
  in which agents with very little computational power and interacting randomly 
  cooperate to collectively perform computational tasks. 
Introduced to model animal populations equipped with sensors~\cite{AADFP06}, 
  they have proved a useful abstraction for settings from wireless sensor networks~\cite{PVV09, DV12}, 
  to gene regulatory networks~\cite{BB04}, and chemical reaction networks~\cite{CCDS15}. 
In this last context, there is an intriguing line of applied research showing that population protocols 
  can be implemented at the level of DNA molecules~\cite{CDSPCSS13}, and 
  that some natural protocols are equivalent to computational tasks solved by living cells 
  in order to function correctly~\cite{CCN12}. 

A population protocol consists of a set of $n$ finite-state agents, interacting 
  in randomly chosen pairs, where each interaction may update the local state of both participants. 
A \emph{configuration} captures the ``global state'' of the system at any given time: 
  since agents are anonymous, the configuration can be entirely described 
  by the number of agents in each state.      
The protocol starts in some valid initial configuration, 
  and defines the outcomes of pairwise interactions.
The goal is to have all agents stabilize to some configuration, 
  representing the output of the computation, such that all future configurations satisfy  
  some predicate over the initial configuration of the system. 
  
In the fundamental \emph{majority} task~\cite{AAE08, PVV09, DV12},
 agents start in one of two input states $A$ and $B$, and must stabilize 
  on a decision as to which state has a higher initial count. 
Another important task is \emph{leader election}~\cite{AAE08le, AG15, DS15}, 
  which requires the system to stabilize to final configurations in which 
  a \emph{single} agent is in a special \emph{leader} state. 
One key complexity measure for algorithms is expected \emph{parallel time}, 
  defined as the number of pairwise interactions until stabilization, 
  divided by $n$, the number of agents. 
The other is \emph{state complexity},
  the number of \emph{distinct states} that an agent can internally represent.

This model leads to non-trivial connections 
  between standard computational models and \emph{natural} computation. 
There is strong evidence to suggest that the cell cycle switch in eukaryotic cells solves 
  an approximate version of majority~\cite{CCN12}, and a 
  three-state population protocol for approximate majority was 
  empirically studied as a model of epigenetic cell memory 
  by nucleosome modification~\cite{DMST07}. The majority task is a key component 
  when simulating register machines via population 
  protocols~\cite{AADFP06, AAE08le, AAE08}.
Thus, it is not surprising that there has been considerable interest 
  in the complexity of majority computation~\cite{AAE08, PVV09, DV12, CCN12, BFKMW16, AAEGR17}.

\paragraph{Complexity Thresholds}
On the lower bound side, a progression of 
  deep technical results~\cite{Doty14, CCDS15} culminated in  
  Doty and Soloveichik~\cite{DS15} showing that  
  \emph{leader election is impossible in sub-linear expected time} for protocols 
  which are restricted to a \emph{constant} number of states per agent.
This result can be extended to majority; 
  in fact,~\cite{AAEGR17} generalized it to show that 
  any protocol for exact majority using $\leq 0.5 \cdot \log \log n$ states must take
   $\Omega ( n / \textnormal{polylog } n )$ expected time, 
  even if the initial discrepancy between 
  the two input states $A$ and $B$ is polylogarithmic in $n$.
The only prior known lower bound was proved in~\cite{AGV15}, 
  showing that sublinear expected time is impossible using 
  at most \emph{four} states per agent. 

The first protocol for exact majority 
  was given by Draief and Vojnovic~\cite{DV12} and by Mertzios et al.~\cite{MNRS14}.
The protocol uses only four states, but needs \emph{linear} expected time 
  to stabilize if the initial discrepancy $\epsilon n$ between the two input states 
  is constant.
Later work~\cite{AGV15} gave the first \emph{poly-logarithmic expected time} protocol 
  for exact majority. 
Unfortunately, this algorithm requires a \emph{linear} in $n$ states per agent. 
Reference~\cite{AAEGR17} reduced the state space to $O( \log^2 n )$, 
  by introducing a state quantization technique. 
Another protocol with $O( \log^2 n)$ states, but better stabilization time
  was recently presented in~\cite{BCER17}.

\paragraph{Summary}
The results described above highlight trade-offs between the stabilization time 
  of a population protocol, and the number of states available at each agent. 
In particular, there is currently still an exponential gap between the best known lower bound, 
  of $\Omega( \log \log n )$ states per agent, 
  and the $O( \log^2 n )$ space used by the best known majority algorithm of~\cite{AAEGR17}.   

\paragraph{Contribution} 
In this paper, we address this gap, by providing tight \emph{logarithmic} upper 
  and lower bounds for majority computation in population protocols. 
For instance, when the discrepancy between the initial counts of 
  majority and minority states is not too high, 
  we show that any algorithm which stabilizes 
  in expected time $O(n^{1-c})$ for $c > 0$ requires $\Omega( \log n )$ states.
We also give a new algorithm using $O( \log n )$ states 
  which stabilizes in expected time $O( \log n \cdot \log{\frac{1}{\epsilon}})$,
  where $\epsilon n$ is the discrepancy between input states.
Notice that $\epsilon$ may depend on $n$ and can range from $1/n$ to $1$,
  corresponding to the majority state having an $\epsilon n$ advantage of
  anywhere from $1$ to $n$ more agents in the initial configuration.
Further, we give a new algorithm for leader election using 
  $O(\log{n})$ states and $O(\log^2{n})$ expected stabilization time.

The fact that the optimal state threshold for this problem 
  is logarithmic may not be entirely surprising. 
However, the techniques we develop to achieve this result are non-trivial, 
  and appear to have implications beyond the majority problem. 
We provide an overview of these techniques below. 

To understand the lower bound, it is useful to contrast 
  it with previous techniques. 
The results of~\cite{DS15, AAEGR17} employ three technical steps.
The first step proves that, from an initial configuration, 
  every algorithm must reach a \emph{dense} configuration, 
  where all states that are expressible by the algorithm are present 
  in large (near-linear) count.
The second step consists of applying a \emph{transition ordering lemma} of~\cite{CCDS15}
  which establishes properties that the state transitions must have in
  order to reduce certain state counts fast from dense configurations.
Finally, these properties are used to perform careful ad-hoc \emph{surgery} arguments 
  to show that any algorithm that stabilizes to a correct output faster than allowed 
  using few states must necessarily have executions in which 
  it stabilizes to the wrong output. 
  
A fundamental barrier to better lower bounds is that the \emph{first step} does not hold 
  for algorithms using, e.g. $O( \sqrt {\log n} )$ states: 
  with such a state space, it is possible to build algorithms that never go 
  through a configuration where all states are expressed in high counts. 
The main contribution of our lower bound is circumventing this challenge. 
We develop a generalization of the transition ordering lemma,
  and a new general surgery technique, which do not require the existence of 
  dense configurations. 

Our lower bound is contingent on basic monotonicity assumptions,
  but requires an additional assumption that we call \emph{output dominance},
  which is satisfied by all known majority algorithms, 
  yet leaving open the possibility
  that some non-standard algorithm might be able to circumvent it
  (however, we think this is unlikely to be the case).
We discuss output dominance in detail in~\sectionref{sec:model}.
Since we eliminate the density requirement, our lower bound technique applies 
  to a significantly more general set of initial configurations than in previous arguments. 
It can also be generalized to other types of predicates, such as equality. 

On the upper bound side, we introduce a new synchronization construct, 
  called a \emph{leaderless phase clock}.
A phase clock is an object which allows agents to have an (approximate) common notion 
  of time, by which they collectively count time in phases 
  of $\Theta(n \log n)$ interactions, with bounded skew. 
The phase clock ensures that all agents will be in the same 
  phase during at least $\Theta(\log n)$ interactions of each agent.

Phase clocks are critical for generic register simulations for population protocols, e.g.~\cite{AAER07}. 
However, they are rarely used in algorithm design, since all known constructions require 
  the existence of a \emph{unique leader}, which is expensive to generate. 
One key innovation behind our algorithm is that it is \emph{leaderless}, as agents maintain the shared 
  clock collectively, without relying on a special leader agent. 
At the implementation level, the phase clock is based on a simple but powerful connection 
  to load balancing by power of two choices, e.g.~\cite{ABKU, berenbrink2006balanced, PTW15}. 

We build on the phase clock to obtain a new space-optimal algorithm for majority, 
  called Phased-Majority. 
In a nutshell, the algorithm splits agents into \emph{workers}, 
  whose job is to compute the majority value, 
  and \emph{clocks}, which implement a leaderless phase clock.
Workers alternate carefully-designed \emph{cancellation} and \emph{doubling} phases.
In the former, agents of disagreeing opinions as to the initial majority 
  cancel each other out, while the latter agents attempt to spread their current opinion.
These dynamics ensure stabilization in 
  $O( \log{n} \cdot \log{\frac{1}{\epsilon}})$ time, 
  both in expectation and with high probability.
Splitting a state space in different types is common, i.e. in ``Leader-Minion'' algorithm 
  of~\cite{AG15} where each state is either a leader or a minion. 
However, doing so explicitly at the beginning of the protocol
  and maintaining a proportion of counts of agents in certain types of states 
  is due to Ghaffari and Parter~\cite{GP16}.
Our cancellation and doubling phases are inspired by~\cite{AAE08le}.

We further exploit the phase clock to obtain a simple algorithm for leader election 
  using $O( \log n )$ states, which stabilizes in $O( \log^2 n )$ expected time.
Prior to this, the best constructions used $O(\log^2 n)$ states~\cite{AAEGR17, BCER17}.
However, based on a different phase clock construction, parallel work 
  by G{\k{a}}sieniec and Stachowiak~\cite{GS17} has designed a 
  polylogarithmic-time leader election protocol using $O(\log \log n)$ states.
This is optimal due to the unified lower bound of~\cite{AAEGR17} for majority and leader election.
Combined, our results and~\cite{GS17} demonstrate an exponential separation 
  between the space complexity of leader election and majority in this model.

\paragraph{Stabilization vs Convergence}
A protocol is said to \emph{converge} to the correct output when its execution first
  reaches a point after which all configurations satisfy the correct output requirement,
  despite possibly non-zero probability of further divergence.
However, a protocol is said to \emph{stabilize} only when the probability of reaching
  a configuration with an incorrect decision actually becomes $0$.
In this paper we exclusively deal with the stabilization requirement.
We should note that~\cite{AAE08le} provides a protocol using a constant number of states
  and with a polylogarithmic expected parallel convergence time 
  if the initial configuration is equipped with a leader.
Our lower bound applies to such initial configurations and demonstrates an interesting separation,
  as for similarly fast stabilization, $\Omega(\log{n})$ states would be necessary.

%% file: model.tex
\section{Model and Problem Statement}
\label{sec:model}

A \emph{task} in the population protocol model is specified by a finite set of input states $I$, and a finite set of output symbols, $O$. 
The predicate corresponding to the task maps any input configuration onto an allowable set of output symbols. We instantiate this definition for majority and leader election below. 

A \emph{population protocol} $\mathcal{P}_k$ with $k$ states
  is defined by a triple $\mathcal{P}_k = (\Lambda_k, \delta_k, \gamma_k)$.
  $\Lambda_k$ is the set of \emph{states} available to the protocol, satisfying $I \subseteq \Lambda_k$ and $|\Lambda_k| = k$.
  The protocol consists of a set of state transitions of the type 
  $A + B \rightarrow C + D,$
  \noindent defined by the protocol's state transition function $\delta_k : \Lambda_k \times \Lambda_k \rightarrow \Lambda_k \times \Lambda_k$. 
	Finally, $\gamma_k : \Lambda_k \rightarrow O$ is the protocol's output function.
This definition extends to protocols which work for \emph{variable} number of states: 
  in that case, the population protocol $\mathcal{P}$ will be a sequence of protocols 
  $\mathcal{P}_i, \mathcal{P}_{i+1}, \ldots$, 
  where $\mathcal{P}_i$ is the protocol with $i$ states.
Later in this section, we will explain in detail how the number of states 
  used by the protocol relates to the number of agents in the system.   

In the following, we will assume a set of $n \geq 2$ agents, interacting pairwise. 
Each agent executes a deterministic state machine, with states in the set $\Lambda_k$. 
The \emph{legal initial configurations} of the protocol are exactly configurations where each agent starts in a state from $I$. 
Once started, each agent keeps updating its state following interactions with other agents, according to a transition function $\delta_k$. 
Each \emph{execution step} is one interaction between a pair of agents, selected to interact uniformly at random from the set of all pairs. 
The agents in states $S_1$ and $S_2$ transition to states 
  given by $\delta_k( S_1, S_2 )$ after the interaction.

\paragraph{Configurations} 
Agents are \emph{anonymous}, 
  so any two agents in the same state are identical and interchangeable. 
Thus, we represent any set of agents simply 
  by the \emph{counts of agents} in every state, which we call a \emph{configuration}.
More formally, a \emph{configuration} $c$ is a function 
  $c: \Lambda_k \to \mathbb{N}$, where $c(S)$ represents the 
  \emph{number of agents in state $S$ in configuration $c$}.
We let $|c|$ stand for the sum, over all states $S \in \Lambda_k$, of $c(S)$,
  which is the same as the total number of agents in configuration $c$.
For instance, if $c$ is a configuration of all agents in the system,
  then $c$ describes the global state of the system, and $|c| = n$.

We say that a configuration $c'$ is \emph{reachable} from a configuration $c$, 
  denoted $c \reach c'$, if there exists a sequence of consecutive steps 
  (interactions from $\delta_k$ between pairs of agents) 
  leading from $c$ to $c'$.
If the transition sequence is $p$, we will also write $c \reach_p c'$.
We call a configuration $c$ the \emph{sum of configurations} $c_1$ and $c_2$  
  and write $c = c_1 + c_2$, when $c(S) = c_1(S) + c_2(S)$ for all states $S \in \Lambda_k$.

\paragraph{The Majority Problem}
In the \emph{majority problem}, agents start in one of two initial states $A, B \in I$.
The output set is $O = \{\id{Win}_A, \id{Win}_B\}$, where, intuitively, an initial state wins if its initial count is larger than the other state's. 
Formally, given an initial configuration $i_n$,
  it is standard to define $\epsilon$ as $\frac{| i_n(A) - i_n(B) |}{n}$.
Thus, $\epsilon$ depends on $n$ and may take values from $1/n$ to $1$.
We will be interested in the value $\epsilon n = | i_n(A) - i_n(B) |$,
  called the \emph{discrepancy}, 
  i.e. initial relative advantage of the majority state.

We say that a configuration $c$ \emph{correctly outputs the majority decision} for $i_n$,
  when for any state $S \in \Lambda_k$ with $c(S) > 0$, 
  if $i_n(A) > i_n(B)$ then $\gamma_k(S) = \id{Win}_A$,
  and if $i_n(B) > i_n(A)$ then $\gamma_k(S) = \id{Win}_B$.
  (The output in case of an initial tie can be arbitrary.) 
A configuration $c$ has a \emph{stable correct majority decision} 
  for $i_n$, if for all configurations $c'$ with $c \reach c'$, 
  $c'$ correctly outputs the majority decision for $i_n$.

In this paper we consider the \emph{exact} majority task, 
  as opposed to \emph{approximate} majority~\cite{AAE08}, 
  which allows agents to produce the wrong output with some probability.
The number of steps until the system reaches a configuration
  with a stable correct majority decision clearly depends on the 
  randomness in selecting interaction partners at each step.
We say that a population protocol $\mathcal{P}_k$ \emph{stably computes majority decision}
  from $i_n$ within $\ell$ steps with probability $1 - \phi$,
  if, with probability $1 - \phi$, any configuration $c$ reachable from $i_n$ by 
  the protocol with $\geq \ell$ steps has a stable correct majority decision
  (with the remaining probability $\phi$, more steps are required in order 
  for the system to stabilize to the correct decision). 

\paragraph{Leader Election} 
In the \emph{leader election} problem, $I = \{A \}$ and 
  in the initial configuration $i_n$ all agents start in the same initial state $A$.
The output set is $O = \{\id{Win}, \id{Lose}\}$. 
Intuitively, a single agent should output $\id{Win}$, 
  while the others should output $\id{Lose}$.

We say that a configuration $c$ \emph{has a single leader} if there exists 
  some state $S \in \Lambda_n$ with $\gamma_n(S) = \id{Win}$ and $c(S) = 1$, 
  such that for any other state $S' \neq S$, $c(S') > 0$ implies 
  $\gamma_n(S') = \id{Lose}$.
A configuration $c$ of $n$ agents has a \emph{stable leader},
  if for all $c'$ reachable from $c$, it holds that $c'$ has a single leader.\footnote{Thi standard definition allows different agents to assume the identity of the single leader after stabilization. We could additionally require that the leader agent remains the same. This is satisfied by our leader election protocol.}

A population protocol $\mathcal{P}_k$ \emph{stably elects a leader}
  within $r$ steps with probability $1 - \phi$, 
  if, with probability $1 - \phi$, any configuration $c$ reachable from $i_n$ by 
  the protocol within $\geq r$ steps has a stable leader.


\paragraph{Complexity Measures} 
The above setup considers sequential interactions; 
  however, interactions between pairs of distinct agents are independent, 
  and are usually considered as occurring in parallel. 
It is customary to define one unit of \emph{parallel time} as $n$ consecutive steps 
  of the protocol.

A population protocol $\mathcal{P}$ stably elects a leader using $s(n)$ states 
  in time $t(n)$ if, for all sufficiently large $n$, the expected number of steps for
  protocol $\mathcal{P}_{s(n)}$ (with $s(n)$ states) to stably elect a leader
  from the initial configuration, divided by $n$, is $t(n)$. 
We call $s(n)$ the \emph{state complexity} and $t(n)$ the \emph{time complexity}
  (or stabilization time) of the protocol.
For the majority problem, the complexity measures might also depend on $\epsilon$.
Thus, $\mathcal{P}$ having state complexity $s(n, \epsilon)$ 
  and time complexity $t(n, \epsilon)$ means that for sufficiently large $n$,
  $\mathcal{P}_{s(n, \epsilon)}$ stabilizes to the correct majority decision
  in expected time $t(n, \epsilon)$ for all $\epsilon$.
If the expected time is finite, then we say that population protocol 
  stably elects a leader (or stably computes majority decision).

\paragraph{Monotonicity}
The above definition of population protocols only requires that for any $n$, 
  there is just one protocol $\mathcal{P}_{s(n)}$ that stabilizes fast for $n$ agents.
In particular, notice that, so far, we did not constrain how protocols $\mathcal{P}_k$ 
  with different number of states $k$ are related to each other.

Additionally, we would like our protocols to be \emph{monotonic}, meaning that 
  a population protocol with a certain number of states that solves a task 
  for $n$ agents should not be slower when running with $n' < n$ agents.
Formally, a monotonic population protocol $\mathcal{P}$ stably elects a leader
  with $s(n)$ states in time $t(n)$, if there exists a sufficiently large constant $d$,
  such that for all $n \geq d$, protocol $\mathcal{P}_{s(n)}$ stably elects a leader 
  from the initial configuration $i_{n'}$ of $n'$ agents,
  for any $n'$ with $d \leq n' \leq n$, in expected parallel time $t(n)$.

A monotonic population protocol $\mathcal{P}$ stably computes majority decision
  with $s(n, \epsilon)$ states in time $t(n, \epsilon)$, 
  if there exists a sufficiently large constant $d$, such that for all $n \geq d$, 
  $\mathcal{P}_{s(n, \epsilon)}$ stably computes majority decision 
  from the initial configuration $i_{n'}$ of $n'$ agents with discrepancy $\epsilon' n'$,
  for any $n'$ with $d \leq n' \leq n$ and $\epsilon' \geq \epsilon$,
  in expected parallel time $t(n, \epsilon)$.
  
\paragraph{Output Dominance} 
Our lower bound will make the following additional assumption on the output properties of population protocols for majority: 

  \begin{definition}[Output Dominance]
For any population protocol $\mathcal{P}_k \in \mathcal{P}$, let 
  $c$ be a configuration with a stable majority decision.  Let
  let $c'$ be another configuration, such that for any state 
  $S \in \Lambda_k$, if $c'(S) > 0$, then $c(S) > 0$.
Then, for any configuration $c''$ such that $c' \reach c''$, 
  if $c''$ has a stable majority decision, then this decision is the same as in $c$.
\end{definition}

Intuitively, output dominance says that, if we change the \emph{counts} of states in any configuration $c$ with a stable output, 
 then the protocol will still stabilize to the same output decision. In other words, the protocol cannot swap output decisions from a stable configuration if the count of some states changes. 
To our knowledge, all known techniques for achieving exact majority in population protocols satisfy this condition. 



%% file: majlower.tex
\section{Lower Bound on Majority}
\begin{theorem}
\label{thm:majmainlb}
Assume any monotonic population protocol $\mathcal{P}$ satisfying output dominance,
  which stably computes majority decision using $s(n, \epsilon)$ states. 
Then, the time complexity of $\mathcal{P}$ must be
  $\Omega\left(\frac{n-2\epsilon n}{3^{2 s(n, \epsilon)} \cdot s(n, \epsilon)^7 \cdot (\epsilon n)^2}\right)$.
\end{theorem}


%
%

\paragraph{Suffix Transition Ordering}
In this section we develop the main technical tool behind the lower bound, called the \emph{suffix transition ordering lemma}.
This result generalizes the classic \emph{transition ordering lemma} 
  of Chen, Cummings, Doty and Soloveichik~\cite{CCDS15},
  that has been a critical piece of the lower bounds of~\cite{DS15, AAEGR17}.
Proofs are deferred to~\sectionref{app:majlower}.

Fix a function $f : \mathbb{N} \to \mathbb{R}^{+}$.
Consider a configuration $c$ reached by an execution of a protocol $\mathcal{P}_k$, 
  and states $r_1, r_2 \in \Lambda_k$. 
A transition $\alpha : (r_1, r_2) \to (z_1, z_2)$ is an $f$\emph{-bottleneck} for $c$,
  if $c(r_1) \cdot c(r_2) \leq f(|c|)$. 
This bottleneck transition implies that the probability of 
  a transition $(r_1, r_2) \to (z_1, z_2)$ is bounded. 
Hence, proving that transition sequences from initial configuration 
  to final configurations contain a bottleneck implies 
  a lower bound on the stabilization time.
Conversely, if a protocol stabilizes fast, then it must be possible 
  to stabilize using a transition sequence which does not contain any bottleneck.
\begin{replemma}{lem:bottlefree}
Consider a population protocol $\mathcal{P}_k$ for majority,
  executing in a system of $n$ agents.
Fix a function $f$. 
Assume that $\mathcal{P}_k$ stabilizes in expected time 
  $o\left(\frac{n}{f(n) \cdot k^2}\right)$ from an initial configuration $i_n$.
Then, for all sufficiently large $n$, there exists a configuration $y_n$ with $n$ agents 
  and a transition sequence $p_n$, such that 
  (1) $i_n \reach_{p_n} y_n$,
  (2) $p_n$ has no $f$-bottleneck, and 
  (3) $y_n$ has a stable majority decision.
\end{replemma}
Next, we prove that, in monotonic population protocols that solve majority,
  the initial state $A$ cannot occur in configurations that have
  a stable majority decision $\id{WIN}_B$, and vice-versa. 
\begin{replemma}{lem:minzero}
Let $\mathcal{P}$ be a monotonic population protocol satisfying output dominance
  that stably computes majority decision for all sufficiently large $n$ 
  using $s(n, \epsilon)$ states.
For all sufficiently large $n'$ and $n > 2n'$, consider executing protocol 
  $\mathcal{P}_{s(n, \epsilon)}$ in a system of $n$ agents, from an initial configuration 
  $i_{n'}$ with $\epsilon n'$ more agents in state $B$.
Consider any $c$ with $i_{n'} \reach c$, that has a stable majority decision $\id{WIN}_B$.
Then $c(A) = 0$.
\end{replemma}
We showed that fast stabilization requires a bottleneck-free transition sequence.
The classic \emph{transition ordering lemma}~\cite{CCDS15}  
  proved that in such a transition sequence, 
  there exists an ordering of all states whose counts decrease more than some threshold, 
  such that, for each of these states $d_j$, 
  the sequence contains at least a certain number of a specific transition that consumes $d_j$, 
  but does not consume or produce any states $d_1, \ldots, d_{j-1}$ that are earlier in the ordering.

A critical prerequisite is proving that counts of states must decrease. 
Towards this goal, for protocols with constant number of states, Doty showed 
  in~\cite{Doty14} that protocols must pass through configurations 
  where all reachable states are in large counts. 
This result was strengthened in~\cite{AAEGR17} to hold for protocols 
  with at most $1/2 \log \log n$ states. 
For protocols with more than $\log \log n$ states, 
  such ``dense'' intermediate configurations may no longer occur. 
Instead, we prove the following \emph{suffix transition ordering lemma},
  which considers an ordering of certain states starting with state $A$, 
  whose count decreases due to~\lemmaref{lem:minzero}.
\begin{replemma}{lem:ordering}[Suffix Transition Ordering Lemma]
Let $\mathcal{P}_k$ be a population protocol 
  executing in a system of $n$ agents.
Fix $b \in \mathbb{N}$, and let $\beta = k^2 b + k b$. 
Let $x, y : \confspacek$ be configurations of $n$ agents such that 
  (1) $x \reach_q y$ via a transition sequence $q$ without a $\beta^2$-bottleneck.
  (2) $x(A) \geq \beta$, and (3) $y(A) = 0$.
Define
\begin{equation*}
\Delta = \{d \in \Lambda_k \mid y(d) \leq b\}
\end{equation*}
to be the set of states whose count in configuration $y$ is at most $b$. 
Then there is an order $\{d_1, d_2,\ldots, d_m\} \subseteq \Delta$, such that $d_1 = A$ and 
  for all $j \in \{1, \ldots, m\}$
  (1) $d_j \in \Delta$, and 
  (2) there is a transition $\alpha_j$ of the form $(d_j, s_j) \rightarrow (o_j, o_j')$
      that occurs at least $b$ times in $q$.
    Moreover, $s_j, o_j, o_j' \in (\Lambda_k \setminus \Delta) \cup \{d_{j+1}, \ldots, d_m\}$.
\end{replemma}
Notice that we do not require the ordering to contain all the states in $\Delta$,
  i.e. we could have $|\Delta| > m$ and $\Delta - \{d_1,\ldots,d_m\} \neq \emptyset$. 
This could happen, for instance, if some state was always present in a zero count.

\paragraph{Proof of Theorem~\ref{thm:majmainlb}}
This technical tool established, we return to the main lower bound proof. 
We will proceed by contradiction. Assume a protocol $\mathcal{P}_{s(n, \epsilon)}$ which would contradict the lower bound. 

Then, for all sufficiently large $n$, $\mathcal{P}_{s(n, \epsilon)}$ stably computes 
  majority decision in expected parallel time 
  $o\left(\frac{n-2\epsilon n}{3^{2 \cdot s(n, \epsilon)} \cdot s(n, \epsilon)^7 \cdot (\epsilon n)^2}\right)$.
We denote $k = s(n, \epsilon)$, $n' = \frac{n-2\epsilon n}{k+1}$,
  $b(n) = 3^k  \cdot (2\epsilon n)$ and $\beta(n) = k^2 \cdot b(n) + k \cdot b(n)$.
Let $i_{n'}$ be an initial configuration of $n'$ agents,
  with $\epsilon n'$ more agents in state $B$.

By monotonicity of the protocol $\mathcal{P}$, $\mathcal{P}_k$ should 
  also stabilize from $i_{n'}$ in expected time 
  $o\left(\frac{n - 2\epsilon n}{3^{2k} \cdot k^7 \cdot (\epsilon n)^2}\right)$,
  which is the same as $o\left(\frac{n'}{k^2 \cdot \beta(n)^2}\right)$.
Thus, by~\lemmaref{lem:bottlefree}, there exists a transition sequence $q$ without
  a $\beta(n)^2$ bottleneck, and configuration $y_{n'}$ with a stable majority decision,
  such that $i_{n'} \reach_q y_{n'}$.

Recall that the discrepancy $\epsilon n$ describes how many more agents
  were in the initial majority state than in the initial minority state,
  and can have a value anywhere from $1$ to $n$.   
The bound is only non-trivial in the case when $\epsilon n \in o(\sqrt{n})$,
  and $n' = \frac{n - 2 \epsilon n}{k+1} \in \omega(k^2 \cdot \beta(n)^2)$.
In this case, we have $i_{n'}(A) = \frac{n' - \epsilon n'}{2} \geq \beta(n)$
  for all sufficiently large $n$.
Also, by~\lemmaref{lem:minzero}, $y_{n'}(A) = 0$.
Therefore, we can apply the suffix transition ordering~\lemmaref{lem:ordering} 
  with $\mathcal{P}_k$, $b = b(n)$ and $\beta = \beta(n)$.
This gives an ordering $\{d_1, \ldots, d_m\}$ on a subset of $\Delta$
  and corresponding transitions $\alpha_j$.
\begin{repclaim}{clm:surgery}
Let $n'' = n' \cdot (m+1) + 2 \epsilon n$ and $i$ be an initial configuration
  of $n''$ agents consisting of $m+1$ copies of configuration $i_{n'}$ plus
  $2 \epsilon n$ agents in state $A$.
  Then, $i \reach z$, for a configuration $z$, such that for all $s \in \Lambda_k$,
  if $z(s) > 0$ then $y_{n'}(s) > 0$.

Here $y_{n'}$ comes from the application of~\lemmaref{lem:bottlefree} and 
  $m$ is the size of the ordering on a subset of $\Delta$.
\end{repclaim}
\begin{proof}[Proof Sketch]
In this proof, we consider transition sequences that might temporarily
  bring counts of agents in certain states below zero.
This will not be a problem because later we add more agents in these states, 
  so that the final transition sequence is well-formed. That is, no count ever falls below zero.  

We proceed by induction, as follows.
For every $j$ with $1 \leq j \leq m$, consider an initial configuration $\iota_j$ 
  consisting of $j$ copies of configuration $i_{n'}$ plus $2 \epsilon n$ agents in state $A$.
Then, there exists a transition sequence $q_j$ from $\iota_j$ 
  that leads to a configuration $z_j$, with the following properties:
\begin{itemize}
\item[1.] For any $d \in \Delta - \{d_{j+1}, \ldots, d_m\}$,
  the count of agents in $d$ remains non-negative throughout $q_j$.
  Moreover, if $y_{n'}(d) = 0$, then $z_j(d) = 0$.
\item[2.] For any $d \not \in \Delta - \{d_{j+1}, \ldots, d_m\}$
    the minimum count of agents in $d$ during $q_j$ is $\geq -3^j \cdot (2 \epsilon n)$.
\item[3.] For any $d \in \{d_{j+1}, \ldots, d_m\}$,
  if $y_{n'}(d) = 0$, then $|z_j(d)| \leq 3^j \cdot (2 \epsilon n)$.
\end{itemize}
\noindent The technical details of this inductive argument are deferred to~\sectionref{app:majlower}.
Given this, we take $i = i_{n'} + \iota_m$ and $z = y_{n'} + z_m$.
The transition sequence $p$ from $i$ to $z$ starts by
  $q$ from $i_{n'}$ to $y_{n'}$, followed by $q_m$.

By the first property of $q_m$, and the fact that no count is ever 
  negative in $q$ from $i_{n'}$ to $y_{n'}$, for any $d \in \Delta$, 
  the count of agents in state $d$ never becomes negative during $p$.
Next, consider any state $d \in \Lambda_k \setminus \Delta$.
By the second property, when $q_m$ is executed from $\iota_m$ to $z_m$, 
  the minimum possible count in $q_m$ is $-3^m \cdot (2 \epsilon n)$.
However, in transition sequence $p$, $q_m$ from $\iota_m$ to $z_m$ follows $q$, 
  and after $q$ we have an extra configuration $y_{n'}$ in the system.
By the definition of $\Delta$,
  $y_{n'}(d) \geq b(n) \geq 3^k \cdot (2 \epsilon n) \geq 3^m \cdot (2 \epsilon n)$.
Therefore, the count of agents in $d$ also never becomes negative
  during $p$, and thus the final transition sequence $p$ is well-formed.

Now, consider a state $s$, such that $y_{n'}(s) = 0$.
We only need to show that $z(s) = 0$.
By definition of $\Delta$, we have $s \in \Delta$, and the first property implies 
  $z(s) = z_m(s) = 0$, completing the proof of the claim.
\end{proof}

Returning to the main thread, we have $n'' \leq n$ due to $m \leq k$.
Moreover, the initial configuration $i$ of $n''$ agents
  has at least $\epsilon n \geq \epsilon n''$ more agents in state $A$ than $B$
  (since $(m+1) \cdot \epsilon n' \leq \epsilon n$, which follows from 
  $(m+1)n' \leq (k+1)n' \leq n$).
So, monotonicity of $\mathcal{P}$ implies that
  $\mathcal{P}_k$ also stably computes majority decision from initial configuration $i$. 
We know $i \reach z$, so it must be possible to reach a configuration
  $y$ from $z$ that has a stable majority decision (otherwise $\mathcal{P}_k$ 
  would not have a finite time complexity to stabilize from $i$).
By output dominance property of $\mathcal{P}$ for $z$ and $y_{n'}$, 
  $y$ has to have the same majority decision as $y_{n'}$.
However, the correct majority decision is $\id{WIN}_B$ in $i_{n'}$
  and $\id{WIN}_A$ in $i$.
This contradiction completes the proof of the theorem. We now make a few remarks on this proof. 

This lower bound implies, for instance, that for $\epsilon = 1/n$, a 
  monotonic protocol satisfying output dominance and stably solves majority
  using $\log{n}/(4 \log{3})$ states, needs to have time 
  complexity $\Omega(\sqrt{n}/\polylog n)$.
But we can get a slightly weaker bound without monotonicity.

\paragraph{Monotonicity}
We use monotonicity of the protocol to invoke the same protocol 
  with different number of agents.
In particular, in~\theoremref{thm:majmainlb}, 
  if the protocol uses $k$ states for $n$ agents,
  we need to be able to use the same protocol for $n/k$ agents.
Suppose instead that the protocol used for more agents never has less states\footnote{Formally, we require that $s(n, \epsilon)$ for any fixed $\epsilon$ be monotonically non-decreasing for all sufficiently large $n$.}.
If the state complexity is $k \leq \log{n}/(2\log\log{n})$,
  then we can find infinitely many $n$ with the desired property that
  the same protocol works for $n/k$ and $n$ agents.
This allows us to apply the same lower bound argument,
  but we would only get a lower bound for state complexities 
  up to $\log{n}/(2 \log \log{n})$.

  


%% file: sync.tex
\section{Leaderless Phase Clock}
\label{sec:sync}

Intuitively, the phase clock works as follows. 
Each agent keeps a local counter, intialized at $0$.
On each interaction, the two agents compare their values, 
  and the one with the \emph{lower} counter value increments its local counter. 
Since interactions are uniformly random, we can connect this to the classic power-of-two-choices load balancing process~\cite{ABKU, PTW15} to obtain 
  that the agents' counter values are concentrated within an additive 
  $O( \log n )$ factor with respect to the mean, with high probability. 

The above procedure has the obvious drawback that, as the counters continue to increment, 
  agents will need unbounded space to store the values. 
We overcome this as follows.
We fix a period $\Psi \in \Theta(\log{n})$, and a range value $\rho \in \Theta(\log{n})$, 
  with $\Psi \gg \rho$. 
The goal of the algorithm is to maintain a ``phase clock" with values between $0$ and $\Psi - 1$, 
  with the property that clocks at different agents are guaranteed to be within some 
  interval of range $\rho$ around the mean clock value, with high probability. 

We let each phase clock state be $V_i$, where $i$ is from $0$ to $\Psi - 1$
  and represents the counter value of the agent in state $V_i$.
The update rule upon each interaction is as follows.
For any $i \leq j$, if $i \not \in [0, \rho - 1]$ or $j \not \in [\Psi - \rho, \Psi - 1]$,
  then we are not dealing with a wrap-around and let
  the agent that has the \emph{lower} counter value increment its local counter.  
Formally, in this case we have that
\begin{eqnarray} 
	V_i + V_j \rightarrow V_{i + 1} + V_j.
\end{eqnarray}

\noindent In the second case, the lower of agent values, say $i$, is in  $[0, \rho - 1]$
  while the other value, $j$, is in $[\Psi - \rho, \Psi - 1]$. 
In this case, we simply increment the level of the agent with the \emph{higher} 
  counter value. 
Formally, when $i \in [0, \rho - 1]$ and $j \in [\Psi - \rho, \Psi - 1]$, 
  we have that 
\begin{eqnarray} 
\label{eqn:second}
	V_i + V_j \rightarrow V_{i} + V_{j + 1}.
\end{eqnarray}

\noindent Finally, if an agent would reach counter value $\Psi$ as the result 
  of the increment, it simply resets to value $V_0$: 
\begin{eqnarray} 
\label{eqn:four}
	V_{\Psi - 1} + V_{\Psi - 1} \rightarrow V_{\Psi - 1} + V_0 \textnormal{ and } V_{i} +  V_{\Psi - 1} \rightarrow V_{i} + V_0, \textnormal{ } \forall i \in [0, \rho - 1].
\end{eqnarray}

\paragraph{Analysis}
We will show that counter values stay concentrated around the mean, 
  so that the difference between the largest and the smallest value will 
  be less than $\rho \in O( \log n )$, with high probability. 
The updates in~\ref{eqn:second}---\ref{eqn:four} allow the algorithm to reset the counter value 
  to $0$ periodically, once the values reach a range where inconsistent wrap-arounds 
  become extremely unlikely. 

For any configuration $c$, let $w_{\ell}(c)$ be the weight of agent $\ell$, 
  defined as follows.
Assume agent $\ell$ is in state $V_i$.
For $i \in [0, \rho-1]$, if in $c$ there exists some agent in state $V_j$ 
  with $j \in [\Psi - \rho, \Psi - 1]$ 
  (i.e. if $\sum_{j \in [\Psi - \rho, \Psi - 1]} c(V_j) > 0$),
  then we have $w_{\ell}(c) = i + \Psi$.
In all other cases, we have $w_\ell(c) = i$.
Given this definition, 
  let $\mu(c) = \frac{\sum_{\ell  =  1}^n w_{\ell}(c)}{n}$ be the mean weight,
  and $x_{\ell}(c) = w_{\ell}(c) - \mu(c)$.
Let us also define $G(c)$, the \emph{gap} in configuration $c$,
  as $\max_{\ell} w_{\ell}(c) - \min_{\ell} w_{\ell}(c)$.
From an initial configuration with a gap sufficiently smaller than $\rho$,
  we consider the number of steps to reach a configuration with a gap of at least $\rho$.
The goal is to show that a large number of steps is required with high probability.
Our definitions are chosen to ensure the following invariant
  as long as the gap is not $\geq \rho$ in the execution:
The evolution of the values $x_{\ell}(c)$ is identical to that of an algorithm where 
  there is no wrap-around once the value would reach $\Psi$. In turn, the algorithm will ensure the gap bound invariant with high probability. 

Therefore, in the following, we will simplify the exposition by considering the process 
  where values continue to increase unboundedly. 
Critically, we notice that this process is now identical to the classical 
  two-choice load-balancing process:
  consider a set of $n$ bins, whose ball counts are initially $0$. 
At each step $t$, we pick two bins uniformly at random, 
  and insert a ball into the \emph{less loaded} of the two. 
Here, let us use $x_{\ell}(t)$ to represents the number of balls in $\ell$-th bin, 
  minus the average number of balls per bin after $t$ steps.
For a fixed constant $\alpha < 1$, define the potential function 
  $\Gamma(t) = \sum_{\ell = 1}^n 2 \cosh (\alpha x_{\ell}(t)) =  \sum_{\ell = 1}^n \left(\exp(\alpha x_{\ell}(t))  + \exp(- \alpha x_{\ell}(t) )\right).$
Peres, Talwar, and Wieder prove the following lemma~\cite{PTW15}. 
\begin{lemma}[Theorem 2.9 in~\cite{PTW15}]
	\label{lem:ptw}
	Given the above process, for any $t \geq 0$, 
$\E [ \Gamma (t + 1) | \Gamma(t) ] \leq \left(  1 - \frac{\alpha}{n}  \right) \Gamma(t) + \theta$, 
where $\alpha < 1$ is a constant from the definition of $\Gamma$ and $\theta \gg 1$ is a fixed constant.
\end{lemma}
\noindent From here, we can prove the following bounded gap property of the leaderless phase clock.
\begin{repcorollary}{cor:ptw}
Suppose $c$ is a configuration with $G(c) \leq \gamma \log n$, for some constant $\gamma$.
Then, for any constant parameter $\beta$, there exists a constant $\gamma'(\beta)$,
  such that with probability $1-m/n^{\beta}$,
  for each configuration $c'$ reached by the $m$ interactions following $c$,
  it holds that $G(c') < \gamma'(\beta) \log{n}$.
\end{repcorollary}

%% file: majupper.tex
\section{Phased Majority}
\paragraph{Overview} At a high level, the state space of the algorithm algorithm is partitioned into  
  into \emph{worker}, \emph{clock}, \emph{backup} and \emph{terminator} states.
Every state falls into one of these categories, 
  allowing us to uniquely categorize the agents based on the state they are in.
The purpose of \emph{worker} agents is to reach a consensus on the output decision.
The purpose of \emph{clock} agents is to synchronize worker agents, enabling a logarithmic state space.
The job of \emph{backup} agents is to ensure correctness via a slower protocol, 
  which is only used with low probability.
The \emph{terminator} agents are there to spread a final majority decision.
Every agent starts as worker, but depending on state transitions, 
  may become a clock, a backup or a terminator.

The algorithm alternates \emph{cancellation} phases, during which workers with different opinions cancel each other out, 
and \emph{doubling} phases, during which workers which still have a ``strong'' opinion attempt to spread it to other agents. 
Clock agents will keep these phases in sync.

\paragraph{State Space} The state of a \emph{worker} agent consists of a triple of: 
  (1) a \emph{phase number} in $\{1, 2, \ldots, 2\log{n}+1\}$; 
  (2) a \emph{value} $\in \{1, 1/2, 0\}$;
  (3) its \emph{current preference} $\id{WIN}_A$ or $\id{WIN}_B$. 
The state of a \emph{clock} agent consists of a pair  
  (1) \emph{position}, a number, 
  describing the current value of its phase clock, initially $0$, and 
  (2) its \emph{current preference} for $\id{WIN}_A$ or $\id{WIN}_B$. 
Backup agents implement a set of four possible states, 
  which serve as a way to implement the four-state protocol of~\cite{DV12,MNRS14}.
We use this as a slow but dependable backup in the case of a low-probability error event. 
There are two terminator states, $D_A$ and $D_B$.
Additionally, every state encodes the agent's original input state ($A$ or $B$) 
  and a single \lit{clock-creation} bit flag.

Agents with input $A$ start in a worker state, with phase number $1$, value $1$, 
  and preference $\id{WIN}_A$.
Agents with input $B$ start in a similar initial state, but with preference $\id{WIN}_B$. 
The \lit{clock-creation} flag is \id{true} for all agents,
  meaning that all agents could still become clocks.
The output of a clock or a worker state is its preference.
The output of an backup state is the output of the
  corresponding state of the $4$-state protocol. 
The output mapping for terminator states is the obvious 
  $\gamma(D_A) = \id{WIN}_A$ and $\gamma(D_B) = \id{WIN}_B$.

A worker agent is \emph{strong} if its current \emph{value} is $1/2$ or $1$. 
A worker agent with value $0$ is \emph{weak}.
We say that a worker is \emph{in phase $\phi$} if its phase number is $\phi$.
For the phase clock, we will set the precise value of the parameter 
  $\rho = \Theta(\log{n})$ in the next section, during the analysis.
The size of the clock will be $\Psi = 4 \rho$.
Clock states with position in $[\rho, 2\rho)$ 
  and $[3\rho, 4\rho)$ will be labelled as buffer states.
We will label states $[0, \rho)$ as \id{ODD} states,
  and $[2\rho, 3\rho)$ as \id{EVEN} states.

We now describe the different interaction types.
Pseudocode is given in~\figureref{fig:phmajcode1} and~\figureref{fig:phmajcode2}.

\paragraph{Backup and Terminator Interactions}
When both agents are backups, 
  they behave as in the $4$-state protocol of~\cite{DV12, MNRS14}.
Backup agents do not change their type,
  but cause non-backup interaction partners to change their type to a backup.
When an agent changes to a backup state, 
  it uses an input state of the $4$-state protocol corresponding to its original input.

An interaction between a terminator agent in state $D_X$ with $X \in \{A, B\}$ 
  and a clock or a worker with preference $\id{WIN}_X$ results in both agents in state $D_X$.
However, both agents end up in backup states 
  after an interaction between $D_A$ and $D_B$, or 
  a terminator agent and a worker/clock agent of the opposite preference.

\paragraph{Clock State Update}
When two clock agents interact, they update positions
  according to the phase clock algorithm described in~\sectionref{sec:sync}.
They might both change to backup states (a low probability event), if their positions 
  had a gap larger than the maximum allowed threshold $\rho$ of the phase clock.
A clock agent that meets a worker agent remains in a clock state 
  with the same position, but adopts the preference of the interaction partner 
  if the interaction partner was strong.

\paragraph{Worker State Update}
Suppose two workers in the same phase interact.
When one is weak and the other is strong,
  the preference of the agent that was weak always 
  gets updated to the preference of the strong agent.

Similar to~\cite{AAE08le}, there are two types of phases.
Odd phases are \emph{cancellation phases},
  and even phases are \emph{doubling phases}.
In a cancellation phase, if both interacting workers have value $1$ but different preferences,
  then both values are updated to $0$, preferences are kept,
  but if \lit{clock-creation} flag is \id{true} at both agents,
  then one of the agents (say, with preference $\id{WIN}_A$) becomes a clock.
Its position is set to $0$ and its preference is carried over from the previous worker state.
This is how clocks are created.
In a doubling phase, if one worker has value $1$ and 
  another has value $0$, then both values are updated to $1/2$.

\paragraph{Worker Phase and State Updates}
Suppose a worker in phase $\phi$ meets a clock.
The clock does not change its state.
If $\phi$ is odd and the label of the clock's state is \id{EVEN},
  or if $\phi$ is even and the label is \id{ODD}, then the worker enters phase $\phi+1$.
Otherwise, the worker does not change its state.

Suppose two workers meet.
If their phase numbers are equal, 
  they interact according to the rules described earlier. 
When one is in phase $\phi$ and another is in phase $\phi+1$,
  the worker in phase $\phi$ enters phase $\phi+1$ (the second worker remains unchanged).
When phase numbers differ by $>1$, both agents become backups.

Here is what happens when a worker enters phase $\phi+1$.
When $\phi+1$ is odd and the agent already had value $1$, then it becomes a
  a terminator in state $D_X$ given its preference was $\id{WIN}_X$ for $X \in \{A, B\}$.
Similarly, if the worker was already in maximum round $\phi = 2 \log{n} + 1$,
  it becomes a terminator with its preference. 
Otherwise, the agent remains a worker and sets phase number to $\phi + 1$.
If $\phi+1$ is odd and the agent had value $1/2$, it updates the value to $1$,
  otherwise, the it keeps the value unchanged.

\paragraph{Clock Creation Flag}
During a cancellation, \lit{clock-creation} flag
  determines whether one of the agents becomes a clock
  instead of becoming a weak worker.
Initially, \lit{clock-creation} is set to \id{true} at every agent.
We will set a threshold $T_c < \rho$, 
  such that when any clock with \lit{clock-creation}=\id{true}
  reaches position $T_c$, it sets \lit{clock-creation} to \id{false}.
During any interaction between two agents, 
  one of which has \lit{clock-creation}=\id{false},
  both agents set \lit{clock-creation} to \id{false}.
An agent can never change \lit{clock-creation} from \id{false} back to \id{true}.

\paragraph{Analysis}
We take a sufficiently large\footnote{For the purposes of~\lemmaref{lem:rumor}, which is given in~\sectionref{app:majupper}.} constant $\beta$, apply~\corollaryref{cor:ptw} with $\gamma = 29(\beta+1)$,
  and take the corresponding $\rho = \gamma'(\beta) \log{n} > \gamma \log{n}$ to be 
  the whp upper bound on the gap that occurs in our phase clock
  (an interaction between two clocks with gap $\geq \rho$ leads to an error and both agents become backups).
We set the \lit{clock-creation} threshold to $T_c = 23(\beta+1)\log{n} < \rho$.
\begin{replemma}{lem:majerror}[Backup]
Let $c$ be a configuration of all agents, containing a backup agent.
Then, within $O(n^2 \log{n})$ expected intaractions from $c$, 
  the system will stabilize to the correct majority decision.
\end{replemma}
We call an execution \emph{backup-free}
  if no agent is ever in a backup state.
Next, we define an invariant and use it to show that the system 
  may never stabilize to the wrong majority decision.
\begin{invariant}[Sum Invariant]
\label{inv:majsum}
Potential $Q(c)$ is defined for configuration $c$ as follows.
For each worker in $c$ in phase $\phi$ with value $v$,
  if its preference is $\id{WIN}_A$, 
  we add $v \cdot 2^{\log{n} - \lfloor (\phi-1)/2 \rfloor}$ to $Q(c)$.
If its preference is $\id{WIN}_B$, we subtract 
  $v \cdot 2^{\log{n} - \lfloor (\phi-1)/2 \rfloor}$ from $Q(c)$.
Suppose $c$ is reachable from an initial configuration
  where input $X \in \{A, B\}$ has the majority with advantage $\epsilon n$,
  by a backup-free execution during which no agent is ever in a terminator state $D_X$.
If $X=A$, we have $Q(c) \geq \epsilon n^2$,
  and if $X=B$, then $Q(c) \leq \epsilon n^2$.
\end{invariant}
\begin{replemma}{lem:majcorrect}[Correctness]
If the protocol stabilizes to $\id{WIN}_X$,
  then $X \in \{A, B\}$ was the initial majority.
\end{replemma}
\begin{replemma}{lem:majdone}[Terminator]
Let $c$ be a configuration of all agents, containing a terminator agent.
In backup-free executions,
  the system stabilizes to the correct majority decision within 
  $O(n \log n)$ interactions in expectation and with high probability.
Otherwise, the system stabilizes within $O(n^2 \log{n})$ expected intaractions.
\end{replemma}
We derive a lemma about each type of phase.
A similar statement is proved for duplication in~\sectionref{app:majupper}.
Since our phases are inspired by~\cite{AAE08le}, here
  we are able to reuse some of their analysis techniques.
\begin{replemma}{lem:majcancel}[Cancellation]
Suppose in configuration $c$ every agent is either a clock 
  or a worker in the same cancellation phase $\phi$ ($\phi$ is odd).
Consider executing $8(\beta+1)n\log{n}$ interactions from $c$
  conditioned on an event that during this interaction sequence,
  no clock is ever in a state with label $\id{EVEN}$,
  and that the phase clock gap is never larger than $\rho$.
Let $c'$ be the resulting configuration.
Then, with probability $1-n^{-\beta}$, in $c'$ it holds that: 
  (1) all strong agents have the same preference,
  or there are at most $n/10$ strong agents with each preference;
  (2) every agent is still a clock, or a worker in phase $\phi$.
\end{replemma}
The final theorem is given below and proved in~\lemmaref{lem:majwhp} 
  and~\lemmaref{lem:majexp} in~\sectionref{app:majupper}.
\begin{theorem}
If the initial majority state has an advantage of $\epsilon n$ agents 
  over the minority state,
  our algorithm stabilizes to the correct majority decision
  in $O(\log{1/\epsilon} \cdot \log{n})$ parallel time, 
  both w.h.p. and in expectation.
\end{theorem}

%% file: appmajlower.tex
\section{Majority Lower Bound}
\label{app:majlower}
Given a protocol $\mathcal{P}_k$ with $k$ states executing in a system of $n$ agents, 
  for a configuration $c$ and a set of configurations $Y$, 
  let us define $T[c \reach Y]$ as the expected parallel time it takes from $c$
  to reach some configuration in $Y$ for the first time.
\begin{lemma}
\label{lem:bottletime}
In a system of $n$ agents executing protocol $\mathcal{P}_k$, 
  let $f : \mathbb{N} \to \mathbb{R}^{+}$ be a fixed function,
  $c : \confspacek$ be a configuration, and 
  $Y$ be a set of configurations, such that every 
  transition sequence from $c$ to some $y \in Y$ has an $f$-bottleneck.
Then it holds that $T[c \reach Y] \geq \frac{n-1}{2 f(n) k^2}$.\footnote{Notice that the assumption is about \emph{every} transition sequence having a bottleneck. Thus, passing some bottleneck cannot be avoided. However, it is true that a particular bottleneck in some fixed configuration can be cleared if the necessary bottleneck states are generated by subsequent non-bottleneck transitions.}  
\end{lemma}
\begin{proof}
By definition, every transition sequence from $c$ to a configuration $y \in Y$ contains 
  an $f$-bottleneck, so it is sufficient to lower bound the expected time for the first 
  $f$-bottleneck transition to occur from $c$ before reaching $Y$.
In any configuration $c'$ reachable from $c$, for any pair of states $r_1, r_2 \in \Lambda_k$ 
  such that $(r_1, r_2) \to (p_1, p_2)$ is an $f$-bottleneck transition in $c'$, 
  the definition implies that $c'(r_1) \cdot c'(r_2) \leq f(n)$.
Thus, the probability that the next pair of agents selected to interact are 
  in states $r_1$ and $r_2$, is at most $\frac{2f(n)}{n(n-1)}$.
Taking an union bound over all $k^2$ possible such transitions, the probability that the 
  next transition is $f$-bottleneck is at most $k^2\frac{2f(n)}{n(n-1)}$.
Bounding by a Bernoulli trial with success probability $\frac{2f(n)k^2}{n(n-1)}$,
  the expected number of interactions until the first $f$-bottleneck transition is at least 
  $\frac{n(n-1)}{2 f(n) k^2}$.
The expected parallel time is this quantity divided by $n$, completing the argument.
\end{proof}
\begin{lemma}
\label{lem:bottlefree}
Consider a population protocol $\mathcal{P}_k$ for majority,
  executing in a system of $n$ agents.
Fix a function $f$. 
Assume that $\mathcal{P}_k$ stabilizes in expected time 
  $o\left(\frac{n}{f(n) \cdot k^2}\right)$ from an initial configuration $i_n$.
Then, for all sufficiently large $n$, there exists a configuration $y_n$ with $n$ agents 
  and a transition sequence $p_n$, such that 
  (1) $i_n \reach_{p_n} y_n$,
  (2) $p_n$ has no $f$-bottleneck, and 
  (3) $y_n$ has a stable majority decision.
\end{lemma}
\begin{proof}
We know that the expected stabilization time from $i_n$ is finite.
Therefore, a configuration $y_n$ that has a stable majority decision
  must be reachable from $i_n$ through some transition sequence $p_n$. 
However, we also need $p_n$ to satisfy the second requirement.

Let $Y_n$ be a set of all stable output configurations with $n$ agents.
Suppose for contradiction that every transition sequence from $i_n$ 
  to some $y \in Y_n$ has an $f$-bottleneck.
Then, using~\lemmaref{lem:bottletime}, the expected time to stabilize from 
  $i_n$ to a majority decision is 
  $T[i_n \reach Y_n] \geq \frac{n-1}{2 f(n) k^2} = \Theta(\frac{n}{f(n)k^2})$.
But we know that the protocol stabilizes from $i_n$ in time $o(\frac{n}{f(n)k^2})$,
  and the contradiction completes the proof.
\end{proof}
\begin{lemma}
\label{lem:minzero}
Let $\mathcal{P}$ be a monotonic population protocol satisfying output dominance
  that stably computes majority decision for all sufficiently large $n$ 
  using $s(n, \epsilon)$ states.
For all sufficiently large $n'$ and $n > 2n'$, consider executing protocol 
  $\mathcal{P}_{s(n, \epsilon)}$ in a system of $n$ agents, from an initial configuration 
  $i_{n'}$ with $\epsilon n'$ more agents in state $B$.
Consider any $c$ with $i_{n'} \reach c$, that has a stable majority decision $\id{WIN}_B$.
Then $c(A) = 0$.
\end{lemma}
\begin{proof}
For sufficiently large $n'$ and $n$,
  we can consider executing protocol $\mathcal{P}_{s(n, \epsilon)}$
  from an initial configuration $i_{n'}$, and know that it stabilizes 
  to the correct majority decision, because $\mathcal{P}$ is a monotonic protocol. 

Assume for contradiction that $c(A) > 0$.
Since $c$ has a stable majority decision $\id{WIN}_B$,
  we must have $\gamma_{s(n, \epsilon)}(A) = \id{WIN}_B$.
Now consider a system of $n$ agents, executing $\mathcal{P}_{s(n, \epsilon)}$,
  where $n'$ agents start in configuration $i_{n'}$ and reach $c$,
  and the remaining agents each start in state $A$.
Clearly, for the system of $n > 2n'$ agents, $A$ is the majority.
Define $c'$ to be configuration $c$ plus $n-n'$ agents in state $A$.
We only added agents in state $A$ from $c$ to $c'$ and $c(A) > 0$, 
  thus for any state $s \in \Lambda_{s(n, \epsilon)}$ with $c'(s) > 0$, we have $c(s) > 0$. 
However, as $c$ has a stable majority $\id{WIN}_B$,
  by output dominance, any configuration $c''$ with $c' \reach c''$
  that has a stable majority decision, should have a decision $\id{WIN}_B$. 

As $\mathcal{P}$ stably computes the majority decision, 
  $\mathcal{P}_{s(n, \epsilon)}$ should stabilize in a finite expected time for $n$ agents.
$c'$ is reachable from an initial configuration of $n$ agents.
Thus, some configuration $c''$ with a stable majority decision must be reachable from $c'$.
However, the initial configuration has majority $A$,
  and $c''$ has a majority decision $\id{WIN}_B$, a contradiction.
\end{proof}
\begin{lemma}
\label{lem:ordering}[Suffix Transition Ordering Lemma]
Let $\mathcal{P}_k$ be a population protocol executing in a system of $n$ agents.
Fix $b \in \mathbb{N}$, and let $\beta = k^2 b + k b$. 
Let $x, y : \confspacek$ be configurations of $n$ agents such that 
  (1) $x \reach_q y$ via a transition sequence $q$ without a $\beta^2$-bottleneck.
  (2) $x(A) \geq \beta$, and (3) $y(A) = 0$.
Define
\begin{equation*}
\Delta = \{d \in \Lambda_k \mid y(d) \leq b\}
\end{equation*}
to be the set of states whose count in configuration $y$ is at most $b$. 
Then there is an order $\{d_1, d_2,\ldots, d_m\} \subseteq \Delta$\footnote{Recall that it is possible that $\{d_1, d_2, \ldots, d_m\} \subset \Delta$.}, such that $d_1 = A$ and 
  for all $j \in \{1, \ldots, m\}$
  (1) $d_j \in \Delta$, and 
  (2) there is a transition $\alpha_j$ of the form $(d_j, s_j) \rightarrow (o_j, o_j')$
      that occurs at least $b$ times in $q$.
    Moreover, $s_j, o_j, o_j' \in (\Lambda_k \setminus \Delta) \cup \{d_{j+1}, \ldots, d_m\}$. 
\end{lemma}
\begin{proof}
We know by definition that $A \in \Delta$.
We will construct the ordering in reverse, i.e. we will determine $e_j$ 
  for $j = |\Delta|, |\Delta|-1, \ldots$ in this order, until $e_j = A$.
Then, we set $m = |\Delta| - j + 1$ and 
  $d_1 = e_j, \ldots, d_m = e_{|\Delta|}$.

We start by setting $j = |\Delta|$.
Let $\Delta_{|\Delta|} = \Delta$.
At each step, we will define the next $\Delta_{j-1}$ as $\Delta_j \setminus \{e_j\}$.
We define $\Phi_j : (\confspacek) \to \mathbb{N}$ based on $\Delta_j$ as 
  $\Phi_j(c) = \sum_{d \in \Delta_j} c(d)$, i.e. the number of agents in states from $\Delta_j$ 
  in configuration $c$.
Notice that once $\Delta_j$ is well-defined, so is $\Phi_j$.

The following works for all $j$ as long as $e_{j'} \neq A$ for all $j' > j$, 
  and thus, lets us construct the ordering.
Because $y(d) \leq b$ for all states in $\Delta$, it follows that 
  $\Phi_j(y) \leq j b \leq k b$.
On the other hand, we know that $x(A) \geq \beta$ and $A \in \Delta_j$, 
  so $\Phi_j(x) \geq \beta \geq k b \geq \Phi_j(y)$.
Let $c'$ be the last configuration along $q$ from $x$ to $y$ where $\Phi_j(c') \geq \beta$,
  and $r$ be the suffix of $q$ after $c'$.
Then, $r$ must contain a subsequence of transitions $u$ each of which strictly decreases $\Phi_j$, 
  with the total decrease over all of $u$ being at least 
  $\Phi_j(c') - \Phi_j(y) \geq \beta - k b \geq k^2 b$.

Let $\alpha: (r_1, r_2) \to (p_1, p_2)$ be any transition in $u$.
$\alpha$ is in $u$ so it strictly decreases $\Phi_j$, and without loss of generality $r_1 \in \Delta_j$.
Transition $\alpha$ is not a $\beta^2$-bottleneck since $q$ does not contain such bottlenecks,
  and all configurations $c$ along $u$ have $c(d) < \beta$ for all $d \in \Delta_j$ by definition of $r$.
Hence, we must have $c(r_2) > \beta$ meaning $r_2 \not\in \Delta_j$.
Exactly one state in $\Delta_j$ decreases its count in transition $\alpha$, 
  but $\alpha$ strictly decreases $\Phi_j$,
  so it must be that both $p_1 \not\in \Delta_j$ and $p_2 \not\in \Delta_j$.
We take $d_j = r_1, s_j = r_2, o_j = p_1$ and $o_j' = p_2$.

There are $k^2$ different types of transitions.
Each transition in $u$ decreases $\Phi_j$ by one and there are at least 
  $k^2 b$ such instances, at least one transition type must repeat 
  in $u$ at least $b$ times, completing the proof.
\end{proof}
\begin{claim}
\label{clm:surgery}
Let $n'' = n' \cdot (m+1) + 2 \epsilon n$ and $i$ be an initial configuration
  of $n''$ agents consisting of $m+1$ copies of configuration $i_{n'}$ plus
  $2 \epsilon n$ agents in state $A$.
Then, $i \reach z$, for a configuration $z$, such that for all $s \in \Lambda_k$,
  if $z(s) > 0$ then $y_{n'}(s) > 0$.

Here $y_{n'}$ comes from the application of~\lemmaref{lem:bottlefree} and 
  $m$ is the size of the ordering on a subset of $\Delta$.
\end{claim}
\begin{proof}
In this proof, we consider transition sequences that might temporarily
  bring counts of agents in certain states below zero.
This will not be a problem because later we add more agents in these states, 
  so that the final transition sequence is well-formed,
  meaning that no count ever falls below zero.  

We do the following induction.
For every $j$ with $1 \leq j \leq m$, consider an initial configuration $\iota_j$ 
  consisting of $j$ copies of configuration $i_{n'}$ plus $2 \epsilon n$ agents in state $A$.
Then, there exists a transition sequence $q_j$ from $\iota_j$ 
  that leads to a configuration $z_j$, with the following properties:
\begin{itemize}
\item[1.] For any $d \in \Delta - \{d_{j+1}, \ldots, d_m\}$,
  the count of agents in $d$ remains non-negative throughout $q_j$.
  Moreover, if $y_{n'}(d) = 0$, then $z_j(d) = 0$.
\item[2.] For any $d \not \in \Delta - \{d_{j+1}, \ldots, d_m\}$
    the minimum count of agents in $d$ during $q_j$ is $\geq -3^j \cdot (2 \epsilon n)$.
\item[3.] For any $d \in \{d_{j+1}, \ldots, d_m\}$,
  if $y_{n'}(d) = 0$, then $|z_j(d)| \leq 3^j \cdot (2 \epsilon n)$.
\end{itemize}

\paragraph{The base case} Consider $j = 1$.
Here $\iota_1$ is simply $i_{n'}$ combined with $2 \epsilon n$ agents in state $A$.
We know $i_{n'} \reach_q y_{n'}$. 
Thus, from $\iota_1$ by the same transition sequence $q$ 
  we reach a configuration $y_{n'}$ plus $2 \epsilon n$ agents in state $d_1 = A$.
Moreover, by suffix transition ordering lemma, we know that transition $\alpha_1$
  of form $(A, s_1) \to (o_1, o_1')$ occurs at least $b(n) \geq (2 \epsilon n)$
  times in $q$.
We add $2 \epsilon n$ occurences of transition $\alpha_1$ at the end of $q$
  and let $q_1$ be the resulting transition sequence.
$z_1$ is the configuration reached by $q_1$ from $\iota_1$.  

For any $d \in \Lambda_k$, during the transition sequence $q$, 
  the counts of agents are non-negative.
In the configuration after $q$, the count of agents in state $d_1 = A$
  is $y_{n'}(A) + 2 \epsilon n = 2 \epsilon n$, and during the 
  remaining transitions of $q_1$ ($2 \epsilon n$ occurences of $\alpha_1$),
  the count of agents in $A$ remains non-negative and 
  reaches $z_1(d_1) = 0$ as required (since $y_{n'}(d_1) = y_{n'}(A) = 0$).
$s_1, o_1, o_1' \in (\Lambda_k \setminus \Delta) \cup \{d_2, \ldots d_m\}$
  implies that for any state $d \in \Delta \setminus \{d_1, d_2, \ldots, d_m\}$,
  the count of agents in $d$ remains unchanged and non-negative 
  for the rest of $q_1$.
Moreover, $z_1(d) = y_{n'}(d)$, thus if $y_{n'}(d) = 0$ then $z_1(d) = 0$.
This completes the proof of the first property.

Now, consider any $d \not\in \Delta - \{d_2, \ldots, d_{m}\}$.
The count of $d$ is non-negative during $q$, 
  and might decrease by at most $2 \epsilon n < 3^j \cdot (2 \epsilon n)$
  during the remaining $2\epsilon n$ occurences of transition $\alpha_1$ in $q_1$
  (achieved only when $s_1 = d$ and $s_1 \neq o_1, o_1'$).
This proves the second property.

The final count of any state in $z_1$ differs by at most 
  $2 \cdot (2 \epsilon n) \leq 3^j \cdot (2 \epsilon n)$
  from the count of the same state in $y_{n'}$.
  (the only states with different counts can be $s_1, o_1$ and $o_1'$, and the largest 
  possible difference of precisely $2 \cdot (2 \epsilon n)$ is attained when $o_1 = o_1'$).
This implies the third property.

\paragraph{Inductive step}
We assume the inductive hypothesis for some $j < m$ and prove it for $j+1$.
Inductive hypothesis gives us configuration $\iota_j$ and a transition sequence $q_j$ 
  to another configuration $z_j$, satisfying the three properties for $j$.
We have $\iota_{j+1} = i_{n'} + \iota_j$, 
  adding another new configuration $i_{n'}$ to previous $\iota_j$. 

Let $u$ be the minimum count of state $d_{j+1}$ during $q_j$.
If $u \geq 0$, we let $q_{j+1}^1 = q$.
Otherwise, we remove $|u| \leq 3^j \cdot (2 \epsilon n) \leq b(n)$ instances
  of transition $\alpha_{j+1}$ from $q$, 
  and call the resulting transition sequence $q_{j+1}^1$.

Now from $\iota_{j+1} = i_{n'} + \iota_j$ consider performing 
  transition sequence $q_{j+1}^1$ followed by $q_j$.
$q_{j+1}^1$ affects the extra configuration $i_{n'}$ 
  (difference between $\iota_j$ and $\iota_{j+1}$),
  and produces $|u|$ extra agents in state $d_{j+1}$ if $u$ was negative.  
Now, when $q_j$ is performed afterwards, 
  the count of state $d_{j+1}$ never becomes negative.  

Let $v$ be the count of $d_{j+1}$ in the configuration reached by 
  the transition sequence $q_{j+1}^1$ followed by $q_j$ from $\iota_{j+1}$.
Since the count never becomes negative, we have $v \geq 0$.
If $y_{n'}(d_{j+1}) > 0$, then we let this sequence be $q_{j+1}$.
If $y_{n'}(d_{j+1}) = 0$, then we add $v$ occurences of transition $\alpha_{j+1}$, 
  i.e. $q_{j+1}$ is $q_{j+1}^1$ followed by $q_j$ 
  followed by $v$ times $\alpha_{j+1}$.
The configuration reached from $\iota_{j+1}$ by $q_{j+1}$ is $z_{j+1}$.

Consider $d \in \Delta - \{d_{j+2}, \ldots, d_m\}$.
For $d = d_{j+1}$, if $y_{n'}(d_{j+1}) = 0$, then we ensured that $z_{j+1}(d_{j+1}) = 0$
  by adding $v$ occurences of transitions $\alpha_{j+1}$ at the end.
In fact, by construction, the count of agents in $d_{j+1}$ 
  never becomes negative during $q_{j+1}$.
It does not become negative during $q_{j+1}^1$ and 
  the $|u|$ extra agents in state $d_{j+1}$ that are introduced 
  ensure futher non-negativity of the count during $q_j$. 
Finally, if the count is positive and $y_{n'}(d_{j+1}) = 0$, 
  it will be reduced to $0$ by the additional occurences of transition $\alpha_{j+1}$, 
  but it will not become negative.
For $d \in \Delta - \{d_{j+1}, d_{j+2}, \ldots, d_m\}$,
  recall that $\alpha_{j+1} = (d_{j+1}, s_{j+1}) \to (o_{j+1}, o_{j+1}')$, where
  $s_{j+1}, o_{j+1}, o_{j+1}' \in (\Lambda_k \setminus \Delta) \cup \{d_{j+2}, \ldots d_m\}$.
Thus, none of $s_{j+1}, o_{j+1}, o_{j+1}'$ are equal to $d$.
This implies that the count of agents in $d$ remain non-negative during $q_{j+1}$
  as the removal and addition of $\alpha_{j+1}$ does not affect the count
  (count is otherwise non-negative during $q$; also during $q_j$ by inductive hypothesis).
If $y_{n'}(d) = 0$, we have $z_{j+1}(d) = z_j(d) + y_{n'}(d) = 0$, as desired.
This proves the first property.

The states for which the minimum count of agents during $q_{j+1}$ 
  might be smaller than during $q_j$ are $s_{j+1}, o_{j+1}$ and $o_{j+1}'$.
Let us first consider $o_{j+1}$ and $o_{j+1}'$.
In our construction, we might have removed at most $3^j \cdot (2 \epsilon n)$ 
  occurences of $\alpha_{j+1}$ from $q$ to get $q_{j+1}^1$, and the largest decrease 
  of count would happen by $2 \cdot 3^j \cdot (2 \epsilon n)$ if $o_{j+1} = o_{j+1}'$.
Adding transitions $\alpha_{j+1}$ at the end only increases the count of $o_{j+1}$ and $o_{j+1}'$. 
Therefore, the minimum count of agents for these two states is 
  $-3^j \cdot (2 \epsilon n) - 2 \cdot 3^j \cdot (2 \epsilon n) = -3^{j+1} \cdot (2 \epsilon n)$,
  as desired.
Now consider state $s_{j+1}$.
We can assume $s_{j+1} \neq o_{j+1}, o_{j+1}'$ as otherwise,
  the counts would either not change or can be analyzed as above for $o_{j+1}$.
Removing occurences of transition $\alpha_{j+1}$ only increases count of $s_{j+1}$,
  and it only decreases if we add $v$ occurences of $\alpha_{j+1}$ at the end to get
  the count of $d_{j+1}$ to $0$.
Since $y_{n'}(d_{j+1})$ should be $0$ in this case in order for us to add transitions at the end,
  we know $v = z_j(d_{j+1})$ if $u \geq 0$, and $v = z_j(d_{j+1}) + |u|$ if $u < 0$. 
In the second case, we remove $|u|$ occurences before adding $v$ occurences,
  so the minimum count in both cases decreases by at most $|z_j(d_{j+1})|$.
By induction hypothesis the minimum count is $\geq -3^j \cdot (2 \epsilon n)$
  and $|z_j(d_{j+1})| \leq 3^j \cdot (2 \epsilon n)$, so the new minimum 
  count of $s_{j+1}$ is $\geq -2 \cdot 3^j \cdot (2 \epsilon n) \geq -3^{j+1} \cdot (2 \epsilon n)$.
This proves the second property.

In order to bound the maximum new $|z_{j+1}(d)|$ for $d \in \{d_{j+2}, \ldots, d_m\}$ with
  $y_{n'}(d) = 0$, we take a similar approach.
Since $y_{n'}(d) = 0$, if $|z_{j+1}(d)|$ differs from $|z_j(d)|$, then $d$ 
  must be either $s_{j+1}, o_{j+1}$ or $o_{j+1}'$.
The minimum negative value that $z_{j+1}(d)$ can achieve can be shown to be
  $3^{j+1} \cdot (2 \epsilon n)$ with the same argument as in the previous paragraph - 
  considering $d = o_{j+1} = o_{j+1}'$ and $d = s_{j+1}$ and estimating 
  the maximum possible decrease, combined with $|z_j(d)| \leq 3^j \cdot (2 \epsilon n)$.
Let us now bound the maximum positive value.
If $d = o_{j+1} = o_{j+1}'$, the increase caused by $v$ additional occurences of $\alpha_{j+1}$
  at the end of $q_{j+1}$ is $2v$.
As before, $v = z_j(d_{j+1})$ if $u \geq 0$, and $v = z_j(d_{j+1}) + |u|$ if $u < 0$,
  and in the second case, we also decrease the count of $d$ by $2 |u|$ 
  when removing $|u|$ occurences of $\alpha_{j+1}$ to build $q_{j+1}^1$ from $q$.
Thus, the maximum increase is $2 |z_j(d_{j+1})| \leq 2 \cdot 3^j \cdot (2 \epsilon n)$.
If $d = s_{j+1}$, then the only increase comes from at most $|u| \leq 3^j \cdot (2 \epsilon n)$
  removed occurences of $\alpha_{j+1}$.
Therefore, the maximum positive value of $z_{j+1}(d)$ equals
  maximum positive value of $z_j(d)$ which is $3^j \cdot (2 \epsilon n)$ plus
  the maximum possible increase of $2 \cdot 3^j \cdot (2 \epsilon n)$,
  giving $3^{j+1} \cdot (2 \epsilon n)$ as desired.
This completes the proof for the third property and of the induction.

\paragraph{The rest of the proof}
We take $i = i_{n'} + \iota_m$ and $z = y_{n'} + z_m$.
The transition sequence $p$ from $i$ to $z$ starts by
  $q$ from $i_{n'}$ to $y_{n'}$, followed by $q_m$.

By the first property of $q_m$, and the fact that no count is ever 
  negative in $q$ from $i_{n'}$ to $y_{n'}$, for any $d \in \Delta$, 
  the count of agents in state $d$ never becomes negative during $p$.
Next, consider any state $d \in \Lambda_k \setminus \Delta$.
By the second property, when $q_m$ is executed from $\iota_m$ to $z_m$, 
  the minimum possible count in $q_m$ is $-3^m \cdot (2 \epsilon n)$.
However, in transition sequence $p$, $q_m$ from $\iota_m$ to $z_m$ follows $q$, 
  and after $q$ we have an extra configuration $y_{n'}$ in the system.
By the definition of $\Delta$,
  $y_{n'}(d) \geq b(n) \geq 3^k \cdot (2 \epsilon n) \geq 3^m \cdot (2 \epsilon n)$.
Therefore, the count of agents in $d$ also never becomes negative
  during $p$, and thus the final transition sequence $p$ is well-formed.

Now, consider a state $s$, such that $y_{n'}(s) = 0$.
We only need to show that $z(s) = 0$.
By definition of $\Delta$, we have $s \in \Delta$, and the first property implies 
  $z(s) = z_m(s) = 0$, completing the proof.
\end{proof}

%% file: appsync.tex
\section{Leaderless Phase Clock}
\begin{corollary}
\label{cor:ptw}
Given the above process, the following holds:
Suppose $c$ is a configuration with $G(c) \leq \gamma \log n$, for some constant $\gamma$.
Then, for any constant parameter $\beta$, there exists a constant $\gamma'(\beta)$,
  such that with probability $1-m/n^{\beta}$,
  for each configuration $c'$ reached by the $m$ interactions following $c$,
  it holds that $G(c') < \gamma'(\beta) \log{n}$.
\end{corollary}
\begin{proof}
We let $\gamma'(\beta) = 2\gamma + \frac{4 + 2\beta}{\alpha}$, 
  where $\alpha$ is the constant from~\lemmaref{lem:ptw},
  and let $\rho = \gamma'(\beta) \log{n}$.
As discussed in~\sectionref{sec:sync}, since we are counting the number of steps from configuration $c$,
  where the gap is less than $\rho$, until the gap becomes $\geq \rho$,
  we can instead analyze the unbounded two-choice process.
In the two choice process, $\Gamma(0)$ corresponds to the potential in configuration $c$.
By simple bounding, we must have that $\Gamma(0) \leq 2 n^{\alpha \gamma + 1}$. 
Assume without loss of generality that $\Gamma(0) = 2 n^{\alpha \gamma + 1}$.

It has already been established by~\lemmaref{lem:ptw} that
$$\E [ \Gamma (t + 1) | \Gamma(t) ] \leq \left(  1 - \frac{\alpha}{n}  \right) \Gamma(t) + \theta.$$
This implies that $\Gamma(t)$ will always tend to \emph{decrease} until it reaches the threshold $\Theta(n)$\footnote{By applying expectation and telescoping, as in the proof of Theorem 2.10 in~\cite{PTW15}.}.
  So, its expectation will always be below its level at step $0$ (in configuration $c$). 
	
Hence, we have that, for any $t \geq 0$, 
$$\E [ \Gamma (t) ] \leq 2 n^{\alpha \gamma + 1}.$$

\noindent By Markov's inequality, we will obtain that 
$$\Pr [ \Gamma (t) \geq n^{\alpha\gamma + 2 + \beta} ] \leq 1/n^{\beta}.$$
	
\noindent It follows by convexity of the exponential and the definition of $\Gamma$ that for each $c'$,
$$\Pr[ G(c') \geq 2(\gamma + (2+\beta)/\alpha) \log{n} ] \leq 1/n^{\beta}.$$
Setting $\rho = \gamma'(\beta) = 2\gamma + \frac{4 + 2\beta}{\alpha}$ and 
  taking union bound over the above event for $m$ steps following configuration $c$ completes the proof.
\end{proof}

%% file: appmajupper.tex
\section{Majority Upper Bound}
\label{app:majupper}
\input{phasemajcode.tex}
\begin{lemma}
\label{lem:clockcnt}
In any reachable configuration of the phased majority algorithm
  from valid initial configurations,
  the number of clock agents is at most $n/2$.
\end{lemma}
\begin{proof}
$n$ workers start in input states and 
  at most one clock is created per two agents in these initial worker states.
This happens only when two workers in the input states with opposite preferences interact
  while \lit{clock-creation} is \id{true}.
However, the values get cancelled, and due to the transition rules, 
  the agent that did not become a clock may never re-enter the initial state.
Therefore, per each clock created there is one agent that will never become a clock,
  proving the claim.
\end{proof}
\begin{lemma}[Rumor Spreading]
\label{lem:rumor}
Suppose that in some configuration $c$, one agent knows a rumor.
The rumor is spread by interactions
  through a set of agents $S$ with $|S| \geq n/2$.
Then, the expected number of interactions from $c$
  for all agents in $S$ to know the rumor is $O(n \log{n})$.
Moreover, for sufficiently large constant $\beta$, 
  after $\beta n \log{n}$ interactions,
  all agents know the rumor with probability $1 - n^{-9}$.
\end{lemma}
\begin{proof}[Proof Adopted]
This problem,
  also known as epidemy spreading, is folklore.
Analysis follows via coupon collector arguments.
The expectation bound is trivial and proved for instance 
  in~\cite{AG15}, Lemma 4.2.

A formal proof of the high probability claim using techniques 
  from~\cite{KMPS95} can for instance be found in~\cite{AAE08le}.
The fact that rumor spreads through at least half of the agents
  affects the bounds by at most a constant factor.
To see this, observe that each interaction has a constant probability of being 
  between agents in $S \cup \{u\}$, where $u$ is the source of the rumor. 
Thus, with high probability by Chernoff, constant fraction of interactions 
  actually occur between these agents and these intaractions act 
  as a rumor spreading on $S \cup \{u\}$.
\end{proof}
\begin{lemma}[Backup]
\label{lem:majerror}
Let $c$ be a configuration of all agents, containing a backup agent.
Then, within $O(n^2 \log{n})$ expected intaractions from $c$, 
  the system will stabilize to the correct majority decision.
\end{lemma}
\begin{proof}
By~\lemmaref{lem:rumor}, within $O(n \log{n})$ expected interactions
  all agents will be in a backup state.
That configuration will correspond to a reachable configuration 
  of the $4$-state protocol of~\cite{DV12, MNRS14}, and
  all remaining interactions will follow this backup protocol.
As the agents have the same input in $4$-state protocol as in the original protocol, 
  it can only stabilize to the correct majority decision.
The $4$-state protocol stabilizes in $n^2 \log{n}$ expected interactions 
  from any reachable configuration, completing the proof.
\end{proof}
\begin{lemma}[Correctness]
\label{lem:majcorrect}
If the system stabilizes to majority decision $\id{WIN}_X$
  for $X \in \{A, B\}$, then state $X$ had the majority 
  in the initial configuration.
\end{lemma}
\begin{proof}
Without loss of generality, assume that state $A$ had the majority 
  in the initial configuration ($\id{WIN}_A$ is the correct decision).
For contradiction, suppose the system stabilizes to the decision $\id{WIN}_B$.
Then, the stable configuration may not contain terminators in state $D_A$ 
  or strong workers with preference $\id{WIN}_A$.
We show that such configurations are unreachable in backup-free executions.

If any agent is in state $D_A$ during the execution,
  it will remain in $D_A$ unless an error occurs (and agents change to backup states).
In neither of these cases can the system stabilize to decision $\id{WIN}_B$.
This is because $\gamma(D_A) = \id{WIN}_A$ and 
  in executions where some agent enters a backup state,
  we stabilize to the correct decision by~\lemmaref{lem:majerror}.

By~\invariantref{inv:majsum}, for any configuration $C$
  reached by a backup-free execution during which, additionally, 
  no agent is ever is state $D_A$, we have $Q(C) \geq n$.
But any configuration $C$ with strictly positive $Q(C)$
  contains at least one strong agent with preference $\id{WIN}_A$, as desired.
\end{proof}
\begin{lemma}[Terminator]
\label{lem:majdone}
Let $c$ be a configuration of all agents, containing a terminator agent.
In backup-free executions,
  the system stabilizes to the correct majority decision within 
  $O(n \log n)$ interactions in expectation and with high probability.
Otherwise, the system stabilizes within $O(n^2 \log{n})$ expected intaractions.
\end{lemma}
\begin{proof}
If there is a backup agent in $c$,
  then the claim follows from~\lemmaref{lem:majerror}.

Otherwise, the terminator spreads the rumor, such that 
  the agents that the rumor has reached are always either in the
  same terminator state, or in an backup state.
By~\lemmaref{lem:rumor}, this takes $O(n \log n)$ interactions
  both in expectation and with high probability.
If all agents are in the same terminator state, 
  then the system has stabilized to the correct majority decision by~\lemmaref{lem:majcorrect}.
Otherwise, there is a backup agent in the system,
  and by~\lemmaref{lem:majerror}, the system will stabilize
  within further $O(n^2 \log{n})$ expected interactions.
\end{proof}
We derive a lemma about each type of phase.
\begin{lemma}[Cancellation]
\label{lem:majcancel}
Suppose in configuration $c$ every agent is either a clock 
  or a worker in the same cancellation phase $\phi$ ($\phi$ is odd).
Consider executing $8(\beta+1)n\log{n}$ interactions from $c$
  conditioned on an event that during this interaction sequence,
  no clock is ever in a state with label $\id{EVEN}$,
  and that the phase clock gap is never larger than $\rho$.
Let $c'$ be the resulting configuration.
Then, with probability $1-n^{-\beta}$, in $c'$ it holds that: 
  (1) all strong agents have the same preference,
  or there are at most $n/10$ strong agents with each preference;
  (2) every agent is still a clock, or a worker in phase $\phi$.
\end{lemma}
\begin{proof}
By our assumption, no clock is ever in a state with 
  label $\id{EVEN}$ during the interaction sequence.
This implies that no worker may enter phase $\phi+1$ or become a terminator.
We assumed that the phase clock gap never violates the threshold $\rho$, 
  and we know all workers are in the same phase, so backups also do not occur.

In configuration $c$, all workers are in phase $\phi$,
  which is a cancellation phase, and must have values in $\{0, 1\}$.
This is true for phase $1$, and 
  when an agent becomes active in a later cancellation phase,
  it updates value $1/2$ to $1$, so having value $1/2$ is impossible. 
Thus, the only strong agents in the system have value $1$.
As no weak worker or a clock may become strong during these 
  $8(\beta+1)n\log{n}$ interactions,
  the count of strong agents never increases.
The only way the count of strong agents decreases is 
  when two agents with value $1$ and opposite preferences interact. 
In this case, the count always decreases by $2$ 
  (both values become $0$ or if \lit{clock-creation}=\id{true}, 
  one agent becomes a clock).

Our claim about the counts then is equivalent to Lemma 5 in~\cite{AAE08le}
  invoked with a different constant ($5$ instead of $4$, as $8(\beta+1)n\log{n} > 5(\beta+1)n\ln{n}$)
  and by treating strong agents with different preferences as $(1,0)$ and $(0,1)$.
\end{proof}
\begin{lemma}[Duplication]
\label{lem:majdup}
Suppose in configuration $c$ every agent is either a clock 
  or a worker in the same duplication phase $\phi$ ($\phi$ is even).
Consider executing $8(\beta+1)n\log{n}$ interactions from $c$
  conditioned on events that during this interaction sequence
  (1) no clock is ever in a state with label $\id{ODD}$,
  (2) the phase clock gap is never larger than $\rho$, and
  (3) the number of weak workers is always $\geq n/10$.
Let $c'$ be the resulting configuration.
Then, with probability $1-n^{-\beta}$, in $c'$ it holds that:
  (1) all strong workers have value $1/2$;
  (2) every agent is still a clock, or a worker in phase $\phi$.
\end{lemma}
\begin{proof}
By our assumption, no clock is ever in a state with 
  label $\id{ODD}$ during the interaction sequence.
This implies that no worker may enter phase $\phi+1$ or become a terminator.
We assumed that the phase clock gap never violates the threshold $\rho$, 
  and we know all workers are in the same phase, so backups also do not occur.

In a duplication phase, workers may not update a state
  such that their value becomes $1$.
Consider a fixed strong worker state in configuration $c$ with value $1$.
By the assumption, probability of an interaction between our fixed agent and a weak
  worker is at least $\frac{n/10}{n(n-1)/2}\geq 1/5n$.
If such an interaction occurs, our agent's value becomes $1/2$.
The probability that this does not happen is at most 
  $(1-1/5n)^{8(\beta+1)n\log{n}} \leq (1-1/5n)^{5n \cdot (\beta+1)\ln{n}} = n^{-\beta-1}$.
By union bound over at most $n$ agents, 
  we get that with probability $1-n^{-\beta}$, 
  no worker will have value $1$, as desired.
\end{proof}
Next, we develop a few more tools before proving stabilization guarantees.
\begin{lemma}
\label{lem:syncinc}
Suppose we execute $\alpha (\beta+1) n \log{n}$ 
  successive interactions for $\alpha \geq 3/2$.
With probability $1-n^{-\beta}$, no agent interacts more than 
  $2 \alpha (1+\sqrt{\frac{3}{2\alpha}}) (\beta+1) \log{n}$ times
  in these interactions.
\end{lemma}
\begin{proof}
Consider a fixed agent in the system.
In any interaction, it has a probability $2/n$ of being chosen.
Thus, we consider a random variable 
  $\mathrm{Bin}(\alpha (\beta+1) n \log{n}, 2/n)$,
  i.e. the number of successes in independent Benoulli trials with probability $2/n$.
By Chernoff bound, setting $\sigma = \sqrt{\frac{3}{2\alpha}} \leq 1$, 
  the probability interacting more than $2\alpha (1+\sigma) (\beta+1) \log{n}$ 
  times is at most $1/n^{\beta+1}$.
Union bound over $n$ agents completes the proof. 

Notice that the number of interactions trivially upper bounds 
  the number of times an agent can go through 
  any type of state transition during these interactions.
In particular, the probability that any clock in the system
  increases its position more than 
  $2 \alpha (1+\sqrt{\frac{3}{2\alpha}}) (\beta+1) \log{n}$ times
  during these interactions is $n^{-\beta}$.
\end{proof}
\begin{lemma}
\label{lem:syncreach}
Consider a configuration in which there are between $2n/5$ and $n/2$ clocks, 
  each with a position in $[0, 2\rho)$, and all remaining agents are workers 
  in the same phase $\phi$, where $\phi$ is odd.
Then, the number of interactions before some clock reaches position 
  $2\rho$ is $O(n \log{n})$ with probability $1-n^{-\beta}$.
\end{lemma}
\begin{proof}
In this case, until some clock reaches position $2\rho$,
  no backup or terminator agents may appear in the system.
Every interaction between two clocks increases one of them.
Therefore, the number of interactions until some clock reaches
  position $2\rho$ is upper bounded by the number of interactions
  until $2\rho n$ interactions are performed between clocks.
At each interaction, two clocks are chosen with probability at least $1/9$
  (for all sufficiently large $n$).
We are interested in the number of Bernoulli trials with 
  success probability $1/9$, necessary to get $2\rho n$
  successes with probability at least $1-n^{-\beta}$.
As we have $\rho = \Theta(\log{n})$, this is $O(n \log{n})$ by Chernoff bound.
\end{proof}
\begin{lemma}
\label{lem:majweak}
Let $\delta(c)$ for a configuration $c$ be the number of weak workers
  minus the number of workers with value $1$.
Suppose that throughout a sequence of interactions from
  configuration $c$ to configuration $c'$ it holds that
  (1) all agents are clocks and workers; and
  (2) no worker enters an odd phase.
Then, $\delta(c') \geq \delta(c)$.
\end{lemma}
\begin{proof}
We will prove that $\delta$ is monotonically non-decreasing 
  for configurations along the interaction sequence from $c$ to $c'$.
Under our assumptions, interactions that affect $\delta$ 
  are cancellations and duplications.
A cancellation decreases the count of workers with value $1$ and increases the 
  count of weak workers, increasing $\delta$ of the configuration.
A duplication decrements both, the number of workers with value $1$,
  and the number of weak workers, leaving $\delta$ unchanged.
\end{proof}
\begin{lemma}
\label{lem:majwhp}
If the initial majority state has an advantage of $\epsilon n$ 
  agents over the minority state,
  our algorithm stabilizes to the correct majority decision
  in $O(\log{1/\epsilon} \cdot \log{n})$ parallel time, with high probability.
\end{lemma}
\begin{proof}
In this argument, we repeatedly consider high probability events,
  and suppose they occur.
In the end, an union bound over all these events gives the desired result.

Consider the first $8(\beta+1)n \log{n}$ interactions of the protocol.
Initially there are no clocks, and each clock starts with a position $0$ and 
  increases its position at most by one per interaction.
By~\lemmaref{lem:syncinc}, with probability $1-n^{-\beta}$, during these interactions 
  no clock may reach position $T_c = 23 (\beta+1) \log{n}$,
  as that would require an agent to interact more than $T_c$ times.
The states of the clock with label \id{EVEN} all have position 
  $2\rho \geq 58(\beta+1) \log{n}$.
Therefore, we can apply~\lemmaref{lem:majcancel} and get that in the 
  resulting configuration $c$, with probability $1-n^{-\beta}$, 
  either all strong workers have the same preference, 
  or the number of strong workers with each preference is at most $n/10$.
We will deal with the case when all strong agents have the same preference later.
For now, suppose the number of strong workers with each preference is at most $n/10$.
As every cancellation up to this point creates one weak worker and one clock, 
  the number of clocks and weak workers is equal and between $2n/5$ and $n/2$.
Thus, for $\delta$ defined as in~\lemmaref{lem:majweak} we have $\delta(c) \geq n/5 > n/10$.
We also know that in configuration $c$, 
  each agent is either a clock that has not yet reached position 
  $T_c = 23 (\beta+1) \log{n}$ 
  (and thus, also not reached a position with a label \id{EVEN}),
  or it is a worker still in phase $1$.

By~\lemmaref{lem:syncreach}, with probability at least $1-n^{-\beta}$,
  within $O(n \log{n})$ interactions we reach a configuration $c'$
  where some clock is at a position $2\rho$, which has a label \id{EVEN}.
But before this, some clock must first reach position $T_c$.
Consider the first configuration $c_1$ when this happens.
The clock at position $T_c$ would set $\lit{clock-creation} \gets \id{false}$.
Notice that from $c_1$, \lit{clock-creation}=\id{false} propagates via
  rumor spreading, and after the rumor reaches all agents, no agent will
  ever have \lit{clock-creation}=\id{true} again, and no more clocks will be created.
By~\lemmaref{lem:rumor}, this will be the case with high probability\footnote{Recall that $\beta$ was chosen precisely to be sufficiently large for the whp claim of~\lemmaref{lem:rumor}.}
  in a configuration $c_2$ reached after $(3/2) \beta n \log{n}$ interactions from $c_1$.
Moreover, by~\lemmaref{lem:syncinc}, no clock will have reached a position larger than
  $T_c + 6(\beta+1)\log{n} \leq 29(\beta+1)\log{n}$ in $c_2$, 
  which is precisely the quantity $\gamma\log{n}$ we used as 
  the maximum starting gap when applying~\corollaryref{cor:ptw} 
  to determine the $\rho$ of our phase clock.
In $c_2$, all clocks have positions in $[0, 29(\beta+1)\log{n})$,
  and no more clocks will ever be created.
By~\lemmaref{lem:clockcnt} and since the number of clocks was $\geq 2n/5$ 
  in configuration $c$, the number of clock agents is from now on 
  fixed between $2n/5$ and $n/2$ (unless some agent becomes a backup or a terminator).
Also, the definition of $\rho$ lets us focus on the high probability 
  event in~\corollaryref{cor:ptw}, that the phase clock gap remains less than 
  $\rho$ during $\Theta(n \log n)$ interactions following $c_2$.

Since $29 (\beta+1) \log{n} < \rho < 2\rho$, in $c_2$ no clock has reached 
  a state with label \id{EVEN}, and thus, configuration $c_2$ occurs 
  after configuration $c$ and before configuration $c'$.
Recall that we reach $c'$ from $c$ 
  within $O(n \log{n})$ interactions with high probability.
In $c'$, some clock has reached position $2\rho$, but the other agents are 
  still either clocks with position in $[\rho, 2\rho)$, or workers in phase $1$.
Let $c''$ be a configuration reached after $(3/2) \beta n \log{n}$ interactions following $c'$.
By~\lemmaref{lem:syncinc}, in $c''$, all clocks will have positions
  $\leq 2\rho + 6(\beta+1)\log{n} < 3\rho$. 
Combining with the fact that at least one agent was at $2\rho$ in $c'$, 
  maximum gap is $<\rho$, and positions $[\rho, 2\rho)$ have label buffer, 
  we obtain that during the $(3/2) \beta n \log{n}$ interactions from $c'$ 
  leading to $c''$, all clocks will be in states with label \id{EVEN} or buffer.
However, there is at least one clock with label \id{EVEN} starting from $c'$,
  spreading the rumor through workers making them enter phase $2$.
Due to~\lemmaref{lem:clockcnt}, at least half of the agents are workers.
Therefore, by~\lemmaref{lem:rumor}, in $c''$, with probability at least 
  $1-n^{-9}$, all worker agents are in phase $2$.
All clocks will be less than gap $\rho$ apart from each other with 
  some clock with a position in $[2\rho, 3\rho)$, 
  and no clock with position $\geq 3\rho$.

We now repeat the argument, but for a duplication phase
  instead of a cancellation using~\lemmaref{lem:majdup},
  and starting with all clocks with positions in $[2\rho, 3\rho)$ 
  as opposed to $[0, \rho)$ and all workers in phase $2$.
We consider a sequence of $8(\beta+1)n \log{n}$ interactions,
  and by~\lemmaref{lem:syncinc}, 
  no clock will reach position $3\rho + 23 (\beta+1)\log{n}$.
Thus, no agent will update to an odd phase and since $\delta(c) \geq n/10$,
  by~\lemmaref{lem:majweak}, the number of weak agents must be at least $n/10$ 
  throughout the interaction sequence, allowing the application of~\lemmaref{lem:majdup}.
We get that with high probability, after $O(n \log{n})$ rounds,
  there will again only be clocks and workers in the system.
All clocks will be less than gap $\rho$ apart with
  some clock at a position in $[3\rho, 0)$ and with 
  no clock yet reaching position $0$ (wrapping around).

Now, due to the loop structure of the phase clock, we can use the 
  same argument as in~\lemmaref{lem:syncreach} to claim that, 
  with probability at least $1-n^{-\beta}$, within $O(n \log{n})$ interactions 
  we reach a configuration where some clock is at a position $0$ (label \id{ODD}).
Because maximum gap is $<\rho$, all clocks will have label buffer,
  and the clock at $0$ will now spread the rumor making all workers enter phase $3$
  within the next $(3/2) \beta n \log{n}$ interactions.
No worker will become a terminator, since~\lemmaref{lem:majdup}
  guarantees that all the agents with value $1$ get their 
  values duplicated (turned into $1/2$) before they enter phase $3$.

Then, we repeat the argument for a cancellation phase (as for phase $1$), 
  except that interactions do not create clock agents 
  (due to \id{clock-creation}=\id{false})
With high probability, within $O(n \log{n})$ interactions,
  all agents will again be in a worker or a clock state.
Moreover, either all strong agents will support the same decision,
  or the number of strong agents supporting each decision will be at most $n/10$.
Since by~\lemmaref{lem:clockcnt}, the number of clocks is at most $n/2$,
  $\delta$ as defined in~\lemmaref{lem:majweak} is 
  at least $n/2 - 2(n/10) - 2(n/10) = n/10$ for this configuration,
  and will remain so until some agent reaches phase $5$,
  allowing us to use~\lemmaref{lem:majdup} for phase $4$, etc.

Due to~\invariantref{inv:majsum}, the case when all
  strong worker agents support the same decision must occur 
  before phase $2\log{1/\epsilon}+1$.
Assume that original majority was $A$, 
  then $Q(c)$ must remain larger than $\epsilon n^2$ 
  (up to this point all agents are clocks or workers, so
  the condition about $D_A$ holds). 
The maximum potential in phase $2\log{1/\epsilon}+1$
  is $\epsilon n^2$ and it is attained when all
  agents are strong and support $\id{WIN}_A$.

Hence, we only need to repeat the argument $O(\log{1/\epsilon})$ times.
The number of high probability events that we did union bound over is 
  $O(n \cdot \log{1/\epsilon} \cdot \log{n})$ 
  (number of interactions for the phase clock).
Combining everything, we get that with probability $1-\frac{O(\log{1/\epsilon})}{n^9}$,
  the algorithm stabilizes within $O(\log{1/\epsilon} \cdot \log{n})$ parallel time.
\end{proof}
\begin{lemma}
\label{lem:majexp}
If the initial majority state has an advantage of $\epsilon n$ agents over 
  the minority state, our algorithm stabilizes to the correct majority decision
  in $O(\log{1/\epsilon} \cdot \log{n})$ expected parallel time.
\end{lemma}
\begin{proof}
We know that in the high probability case of~\lemmaref{lem:majwhp},
  the protocol stabilizes within $O(\log{1/\epsilon} \cdot \log{n})$ parallel time.
What remains to bound the expectation the low probability events 
  of~\lemmaref{lem:majwhp}.

Notice that as soon as any agent gets into an backup or a terminator state,
  by~\lemmaref{lem:majerror} and~\lemmaref{lem:majdone}, 
  the remaining expected time for the protocol to stabilize is $O(n^2 \log{n})$ interactions.
Therefore, we will be looking to bound expected time to reach configurations 
  with a backup or a terminator agent.

Without loss of generality, suppose $A$ is the inital majority.
If all agents start in $A$, then the system is already stable with the correct decision.
If the initial configuration contains just a single agent in state $B$,
  then it takes expected $O(n)$ interactions for this agent to interact with an agent 
  in state $A$, and lead to a configuration where $n-2$ agents are 
  in state $A$ (worker state with value $1$ and preference $\id{WIN}_A$),
  one agent is a worker with value $0$ and one agent is a clock with position $0$.
One of these two agents (weak worker and the clock) has preference 
  $\id{WIN}_B$ and it takes another $O(n)$ expected interactions for it to meet 
  a strong agent with preference $\id{WIN}_A$ and update its own preference.
At that point (after $O(1)$ expected parallel time) the system will be stable 
  with the correct majority decision (since there is only one clock, its position remains at $0$,
  and because of this, workers do not perform any phase updates).

Next, we consider the case when there are at least $2$ agents in state $B$
  in the initial configuration.
Interactions between two agents both in state $A$ and two agents both in state $B$
  do not lead to state updates.
After one cancellation, as in the previous case, there will be agents in input states,
  one clock stuck at position $0$, and one weak worker that might change its preference,
  but not phase or value.
Therefore, after $O(n)$ expected interactions, we will get at least
  two clock agents in the system.

Unless some agent ends up in a backup or a terminator state (this is a good case,
  as discussed earlier) the number of clocks never decreases.
During interactions when there are $k \geq 2$ clocks in the system,
  the probability of an interaction between two clocks
  is $\frac{k(k-1)/2}{n(n-1)/2} \geq k/n^2$.
Therefore, it takes $O(n^2/k)$ expected interactions for one of the 
  clocks to increment its position.
After $k \cdot 4\rho = O(k \log{n})$ such increments of some clock position,
  at least one of the clocks should go through all the possible positions.
Notice that this statement is true without the assumption about the maximum gap 
  of the clock (important, because that was a with high probability 
  guarantee, while here we are deriving an expectation bound
  that holds from all configurations)

Consider any non-clock agent $v$ in the system in some configuration $c$.
Since we know how to deal with the case when some agent ends up in a backup
  or a terminator state, suppose $v$ is a worker.
The clock agent that traverses all positions in $[0, 4\rho)$ necessarily
  passes through a state with label \id{ODD} and with label \id{EVEN}.
If $v$ is in an odd phase and does not move to an even phase,
  then when the clock is in state labelled \id{EVEN}, there would
  be $1/n^2$ chance of interacting with $v$, and vice versa.
If such intaraction occurs, and $v$ does not change its 
  state to a non-worker, then it must necessarily increase its phase.
Therefore, in any given configuration, for any given worker,
  the expected number of interactions before it either changes 
  to a non-worker state or increases it phase is 
  $O(k \log{n} \cdot \frac{n^2}{k} \cdot n^2) = O(n^4 \log{n})$.

By~\lemmaref{lem:clockcnt}, 
  there can be at most $n/2$ clocks in the system in any configuration.
Also, non-worker states can never become worker states again.
The maximum number of times a worker can increase its phase is $O(\log{n})$.
Thus, within $O(n^5 \log^2{n})$ expected interactions,
  either some agent should be in a backup or terminator state,
  or in the maximum phase possible ($2\log{n} + 1$).

If some worker reaches a maximum phase possible, 
  there are no backup or terminator agents
  and there exists another worker with a smaller phase, 
  within $O(n^2)$ expected interactions they will interact.
This will either turn both agents into backups,
  or the other agent will also enter phase $2\log{n}+1$.
Thus, within at most $O(n^3)$ additional expected interactions,
  all workers will be in phase $2\log{n} + 1$ 
  (unless there is a backup or a terminator in the system).
This contradicts with~\invariantref{inv:majsum}, implying that our 
  assumption that no agent gets into a backup or a terminator state 
  should be violated within expected $O(n^5 \log^2{n})$ interactions 
  (using linearity of expectation and discarding asymptotically dominated terms).
Hence, the protocol always stabilizes within $O(n^4 \log^2{n})$ expected parallel time.
The system stabilizes in this expected time in the low probability
  event of~\lemmaref{lem:majwhp}, giving the total expectated time of at most 
  $O(\log{1/\epsilon} \cdot \log{n}) + \frac{O(\log{1/\epsilon} \cdot n^4 \cdot \log^2{n})}{n^9} = O(\log{1/\epsilon} \cdot \log{n})$ as desired.
\end{proof}

%% file: phasemajcode.tex
\begin{figure}[ht]
\hrule
\DontPrintSemicolon
{\centering
{\small
\begin{algorithm}[H]
\SetKwInput{KwState}{Parameters}
\KwState{\;
$\rho,$ an integer $>0$, set to $\Theta(\log{n})$\;
$T_c < \rho,$ an integer $>0$, threshold for \lit{clock-creation}\;
}
\SetKwInput{KwState}{State Space}
\KwState{\;
\begin{tabular}{l l}
$\id{Worker States} =$ & $\id{.phase} \in \{ 1, \ldots, 2\log{n} + 1 \}$, \\
                       & $\id{.value} \in \{0, 1/2, 1\}$ \\
                       & $\id{.preference} \in \{\id{WIN}_A, \id{WIN}_B\}$; \\
$\id{Clock States} =$  & $\id{.position} \in \{ 0, \Psi - 1 \}$,\\
                       & $\id{.preference} \in \{\id{WIN}_A, \id{WIN}_B\}$; \\
$\id{Backup States} =$ & 4 states from the protocol of~\cite{DV12}; \\
$\id{Terminator States} = $ & $\{ D_A, D_B \}$.\\
\end{tabular}
\newline
Additional two bit-flags in every state
\begin{tabular}{l l}
                       & $\id{.InitialState} \in \{A, B\}$ \\
                       & $\id{.clock-creation} \in \{\id{true}, \id{false}\}$; \\
\end{tabular}
}

\KwIn{States of two agents, $S_1$ and $S_2$}
\KwOut{Updated states $S_1' = \lit{update}(S_1, S_2)$ and $S_2' = \lit{update}(S_2, S_1)$}

\SetKwInput{KwState}{Auxiliary Procedures}

\KwState{\;
$\id{backup}( S ) = \left\{ 
 \begin{array}{ll} 
  A_{\textnormal{\cite{DV12}}} & \textnormal{if $S.\id{InitialState} = A$; } \\
  B_{\textnormal{\cite{DV12}}} & \textnormal{otherwise.}
  \end{array} 
  \right. $

$\id{term-preference}( S ) = \left\{ 
 \begin{array}{ll} 
  D_A          & \textnormal{if } S = D_A \textbf{ or } S.\id{preference} = \id{WIN}_A \\
  D_B          & \textnormal{if } S = D_B \textbf{ or } S.\id{preference} = \id{WIN}_B \\
  \end{array} 
  \right. $

$\id{pref-conflict}( S, O ) = \left\{ 
 \begin{array}{ll} 
  \id{true}        & \id{term-preference}(S) \neq \id{term-preference}(O) \\
  \id{false}       & \textnormal{otherwise.}
  \end{array} 
  \right. $

$\id{is-strong}( S ) = \left\{ 
 \begin{array}{ll} 
  \id{true}    & \textnormal{if } S \in \id{Worker States} \textbf{ and } S.\id{value} \neq 0 \\
  \id{false}   & \textnormal{otherwise.}
  \end{array} 
  \right. $

$\id{clock-label}( O ) = \left\{ 
 \begin{array}{ll} 
  0          & \textnormal{if } O.\id{position} \in [2\rho, 3\rho) \\
  1          & \textnormal{if } O.\id{position} \in [0, \rho) \\
  -1        & \textnormal{otherwise.}
  \end{array} 
  \right. $

$\id{inc-phase}( \phi, O) = \left\{ 
 \begin{array}{ll} 
  \id{true}    & \textnormal{if } \phi=O.\id{phase}-1 
                 \textbf{ or } \phi \textnormal{ \emph{mod} } 2 = 1 - \id{clock-label}(O) \\
  \id{false}   & \textnormal{otherwise.}
  \end{array} 
  \right. $
}

\BlankLine
\textbf{procedure} $\lit{update}\langle S, O \rangle$\;
{
\Indp

\If{$S \in \id{Backup States}$ \textbf{or} $O \in \id{Backup States}$}
{
  \If{$S \in \id{Backup States}$ \textbf{and} $O \in \id{Backup States}$}
  {
     $S' \gets \lit{update}_{\textnormal{\cite{DV12}}}(S, O)$ \nllabel{line:4backup} 
  } \uElseIf{$O \in \id{Backup States}$}
  {
    $S' \gets \id{backup}(S)$ \nllabel{line:turnbackup}
  } \lElse{$S' \gets S$}

  \Return $S'$
} 

\tcp{Backup states processed, below $S$ and $O$ are not in backup states}

\If{$S \in \id{Terminator States}$ \textbf{or} $O \in \id{Terminator States}$}
{
  \If{$\id{pref-conflict}(S, O) = \id{false}$}
  {
    $S' \gets \id{term-preference}(S)$
  } \lElse{$S' \gets \id{backup}(S)$}
  \Return $S'$
} 

\Indm
}
\end{algorithm}}}
\hrule
\caption{Pseudocode for the phased majority algorithm, part 1/2}
\label{fig:phmajcode1}
\end{figure}

\begin{figure}[ht]
\hrule
\DontPrintSemicolon
{\centering
{\small
\begin{algorithm}[H]
\setcounter{AlgoLine}{20}
\tcp{Below, both $S$ and $O$ are workers or clocks}

$S' \gets S$

\If{$O.\id{clock-creation} = \id{false}$}
{
  $S'.\id{clock-creation} \gets \id{false}$
}

\If{$\id{is-strong}(S) = \id{false}$ \textbf{and} $\id{is-strong}(O) = \id{true}$}
{
  $S'.\id{preference} \gets O.\id{preference}$
}

\tcp{Clock creation flag and preference updated (always)}

\If{$S \in \id{Clock States}$}
{
  \If{$O \in \id{Clock States}$}
  {
    \tcc{Update $S'.\id{Position}$ according to~\sectionref{sec:sync}. If gap between $S.\id{position}$ and $O.\id{position}$ not less than $\rho$, set $S' \gets \id{backup}(S)$. If $S.\id{position} \geq T_c$, set $S'.\id{clock-creation} \gets \id{false}$.}
  } 
  \Return $S'$
}

\tcp{Below, $S$ is a worker and $O$ is a worker or a clock}

$\phi \gets S.\id{phase}$

\If{$\id{inc-phase}(\phi, O) = \id{true}$}
{
  \If{$\phi = 2\log{n} + 1$ \textbf{ or } ($\phi \textnormal{ \emph{mod} } 2 = 0$ \textbf{ and } $S.\id{value} = 1$)}
  {
    $S' \gets \id{term-preference}(S)$
  } \Else {
    $S.\id{phase} = \phi + 1$

    \If{$\phi \textnormal{ \emph{mod} } 2 = 0$ \textbf{ and } $S.\id{value} = 1/2$}
    {
      $S'.\id{value} = 1$
    }
  }
  \Return $S'$
}

\If{$O \in \id{Clock States}$}
{
  \Return $S'$
}
\tcp{Below, $S$ is a worker and $O$ is a worker}

\If{$|S.\id{phase} - O.\id{phase}| > 1$}
{
  $S' \gets \id{backup}(S)$
  
  \Return $S'$
}

\tcp{Below, worker meets worker within the same phase}
\If{$\phi \textnormal{ \emph{mod} } 2 = 1$}
{
  \tcp{Cancellation phase}
  \If{$S.\id{value} = 1$ \textbf{ and } $O.\id{value} = 1$ \textbf{ and } $\id{pref-conflict}(S, O) = \id{true}$}
  {
    \If{$S'.\id{clock-creation} = true$ \textbf{ and } $S.\id{preference} = \id{WIN}_A$}
    {
      $S' \gets \id{clock}(.\id{position} = 0, .\id{preference} = S.\id{preference})$
    } \Else {$S'.\id{value} \gets 0$}
  }
} \Else {
  \tcp{Doubling phase}
  \If{$S.\id{value} + O.\id{value} = 1$}
  {
    $S'.\id{value} = 1/2$
  }
}

\Return $S'$

\end{algorithm}}}
\hrule
\caption{Pseudocode for the phased majority algorithm, part 2/2}
\label{fig:phmajcode2}
\end{figure}

%% file: leupper.tex
\section{Phased Leader Election}
\paragraph{Overview} We partition the state space into 
  \emph{clock} states, 
  \emph{contender} states, and
  \emph{follower} states.
A clock state is just a position on the phase clock loop.
A contender state and a follower state share the following two fields
  (1) a \emph{phase number} in $[1, m]$, which we will fix to $m=O(\log{n})$ later, and
  (2) a \id{High}/\id{Low} indicator within the phase.
Finally, all states have the following bit flags
  (1) \lit{clock-creation}, as in the majority protocol, and
  (2) a coin bit for generating synthetic coin flips with small bias, 
      as in~\cite{AAEGR17}\footnote{State transitions in population protocols are deterministic and the protocol does not have access to a random coin flips. The idea for synthetic coin flips is to simulate an outcome of a random coin flip based on part of the state of the interaction partner. This can be made to work because the scheduler is randomized.} 
The loop size of the phase clock will be $\Theta(\log{n})$ as in the majority.
Thus, the state complexity of the algorithm is $\Theta(\log{n})$.

All agents start as contenders,
  with phase number $1$ and a \id{High} indicator.
The coin is initialized with $0$ and \lit{clock-creation}=\id{true}.
Each agent flips its coin at every interaction.
As in majority, labels \emph{buffer}, \id{ODD} and \id{EVEN}
  are assigned to clock positions.
Only contenders map to the leader output. 

\paragraph{Clock States and Flags}
Clock agents, as in the majority algorithm, follow the phase clock protocol 
  from~\sectionref{sec:sync} to update their position.
When a clock with \lit{clock-creation}=\id{true} reaches 
  the threshold $T_c$, it sets \lit{clock-creation} to \id{false}.
The \lit{clock-creation} flag works exactly as in the majority protocol.

\paragraph{Contenders and Followers}
The idea of followers that help contenders eliminate each other comes from~\cite{AG15}.
A follower maintains a maximum pair of (phase number, \id{High}/\id{Low} indicator)
  ever encountered in any interaction partner, contender or follower
  (lexicographically ordered, \id{High} $>$ \id{Low}).
When a contender meets another agent with a larger
  phase-indicator pair than its own, it becomes a follower and adopts the pair.
An agent with a strictly larger pair than its interaction partner
  does not update its state/pair.
Also, when two agents with the same pair interact and one of them is a follower,
  both remain in their respective states.

When two contenders with the same pair interact and
  \lit{clock-creation}=\id{true}, one of them becomes a clock at position $0$.
If \lit{clock-creation}=\id{false}, then one of them becomes 
  a follower with the same pair.
The other contender remains in the same state.
As in phased majority, we want to control the counts of states and 
  in particular, avoid creating more than $n/2$ clocks.
This can be accomplished by adding a single \id{created} bit initialized to $0$.
When two contenders with the same pair meet, 
  and both of their \id{created} bit is $0$, then one of them becomes a clock
  and another sets \id{created} to $1$.
Otherwise, if one of the contenders has \id{created}$=1$, then it becomes a
  follower; the other remains unchanged.
Then~\lemmaref{lem:clockcnt} still works 
  and gives that we will never have more than $n/2$ clocks.

\paragraph{Contender Phase Update}
Consider a contender in phase $\phi$.
If $\phi$ is odd phase and the contender meets a clock whose
  state has an \id{EVEN} label, or when $\phi$ is even and the contender
  meets a clock with an \id{ODD}-labelled state,
  then it increments its phase number to $\phi+1$.
However, again due to technical reasons (to guarantee unbiased synthetic randomness),
  entering the next phase happens in two steps.
First the agent changes to a special \emph{intermediate} state 
  (this can be implemented by a single bit that is true if the state is intermediate),
  and only after the next interaction changes to non-intermediate contender with phase $\phi+1$
  and sets the \id{High}/\id{Low} indicator to the coin value of the latest interaction partner.
If the coin was $1$, indicator is set to \id{High}
  and if the coin was $0$, then it is set to \id{Low}.
For the partner, meeting with an intermediate state is almost like missing an interaction -
  only the coin value is flipped.
An exception to the rule of incrementing the phase is obviously when a contender 
  is in phase $m$.
Then the state does not change.
\begin{theorem}
\label{thm:lemain}
Our algorithm elects a unique stable leader within
  $O(\log^2{n})$ parallel time, both with high probability and in expectation.
\end{theorem}
\begin{proof}
We first prove that is always at least one contender in the system.
Assume the contrary, and consider the interaction sequence leading to a contenderless configuration.
Consider the contender which had the highest phase-indicator pair when it got eliminated,
  breaking ties in favor of the later interaction.
This is a contradiction, because no follower or other contender may have
  eliminated it, as this requires having a contender with a larger phase-indicator pair. 

By construction, the \emph{interacted} bit combined with~\lemmaref{lem:clockcnt}
  ensures that there are never more than $n/2$ clocks in the system. 
We set up the phase clock with the same $\rho$ as in majority, and  
  also the \lit{clock-creation} threshold $T_c = 23 (\beta+1) \log{n}$.
After the first $8 (\beta+1) n \log{n}$ interactions,
  with probability $1-n^{-\beta}$, there will be at least $2/5n$ clocks.
The proof of this claim is similar to~\lemmaref{lem:majdup}: if the number of
  contenders with initial state and created set to $0$ was at least $n/10$
  throughout the sequence of interactions, then any given agent
  would have interacted with such agent with high probabiliy, 
  increasing the number of clocks.
Otherwise, the number of agents with $\id{created}=0$ falls under $n/10$, but
  there are as many agents that are clocks as contenders that are not $\id{created}=0$
  and at least $(n - n/10) / 2 > 2n/5$.

Now we can apply the same argument as in~\lemmaref{lem:majwhp} and get that,
  with high probability, the agents will keep entering larger and larger phases.
In each phase, as in the majority argument, 
  a rumor started at each agent reaches all other agents with high probability.
This means that if a contender in a phase selects indicator \id{High},
  then all other contenders that select indicator \id{Low} in the same phase
  will get eliminated with high probability.
By Theorem 4.1 from~\cite{AAEGR17}, the probability that a given contender
  picks \id{High} is at least $1/2 - 1/2^8$ 
  with probability at least $1-2\exp(-\sqrt{n}/4)$.
For every other agent, the probability of choosing $\id{Low}$ is similarly lower bounded.
Thus, Markov's inequality implies that in each phase, the number of contenders
  decreases by a constant fraction with constant probability,
  and phases are independent of each other.
By a  Chernoff bound, it is sufficient to take logarithmically many phases 
  to guarantee that one contender will remain, with high probability, 
  taking a union bound with the event that each phase takes $O(\log{n})$ parallel time,
  as proved in~\lemmaref{lem:majwhp}.

To get an expected bound, observe that when there are more than two contenders 
  in the system, there is $1/n^2$ probability of their meeting.
Hence, the protocol stabilizes from any configuration, in particular
  in the with low probability event, within $O(n^3)$ interactions, which does 
  not affect the total expected parallel time of $O(\log^2{n})$.
\end{proof}